\newif\ifarxiv
\arxivtrue

\ifarxiv
\documentclass{article}
\else
\documentclass[preprint,12pt]{elsarticle}
\fi

\usepackage{url}
\usepackage{times}
\usepackage{xcolor}
\usepackage{soul}
\usepackage{lineno}
\usepackage[utf8]{inputenc}
\usepackage[small]{caption}

\usepackage{tikz}
\usepackage{nicefrac}
\usepackage{amsmath,amsfonts,amsthm}

\ifarxiv
\else
\usepackage{natbib}
\setcitestyle{numbers,sort&compress}
\fi

\newcommand{\pref}{\succ}
\newcommand{\calR}{\mathcal{R}}

\newcommand{\probDef}[3]{\begin{quote}\textsc{\textsc{#1}}\\ \textbf{Input:} #2\\ \textbf{Question:} #3\end{quote}}

\newcommand{\N}{\mathbb{N}}
\newcommand{\Z}{\mathbb{Z}}
\newcommand{\shortcite}{\cite}

\newtheorem{theorem}{Theorem}
\newtheorem{corollary}{Corollary}
\newtheorem{proposition}{Proposition}
\newtheorem{example}{Example}

\theoremstyle{definition}
\newtheorem{definition}{Definition}
\newtheorem{remark}{Remark}

\newcommand{\np}{{\mathrm{NP}}}

\ifarxiv
\else
\journal{Journal of Artificial Intelligence Research}
\fi

\begin{document}

\ifarxiv
\title{Opinion Diffusion and Campaigning on \\ Society Graphs}
\author{
Piotr Faliszewski \\
\small{AGH University} \\
\small{\url{faliszew@agh.edu.pl}}
\and
Rica Gonen \\
\small{The Open University of Israel} \\
\small{\url{ricagonen@gmail.com}}
\and
Martin Koutecký \\
\small{Charles University} \\
\small{\url{koutecky@kam.mff.cuni.cz}}
\and
Nimrod Talmon \\
\small{Ben-Gurion University} \\ \small{\url{talmonn@bgu.ac.il}}
}

\date{}
\maketitle 
\else
\begin{frontmatter}

\title{Opinion Diffusion and Campaigning on Society Graphs\footnote{%
  A preliminary version of this paper was presented at the 27th International Joint Conference on Artificial Intelligence (IJCAI '18)~\cite{preliminary}. This full version contains extended discussions, more elaborate examples, formal treatment of certain extensions, a somewhat stronger result regarding NP-hardness of the problem (including fixing a bug from the preliminary version), and reports on simulations.}}

\author{Piotr Faliszewski}
\address{AGH University \\ \url{faliszew@agh.edu.pl}}

\author{Rica Gonen}
\address{The Open University of Israel \\
\url{ricagonen@gmail.com}}

\author{Martin Koutecký}
\address{Charles University \\
\url{koutecky@kam.mff.cuni.cz}}

\author{Nimrod Talmon}
\address{Ben-Gurion University \\ \url{talmonn@bgu.ac.il}}
\fi

\begin{abstract}
We study the effects of campaigning,
where the society is partitioned into voter clusters
and a diffusion process propagates opinions in a network connecting the clusters.
Our model is very powerful and can incorporate many campaigning actions,
various partitions of the society into clusters,
and very general diffusion processes.
Perhaps surprisingly, we show that computing the cheapest campaign for rigging a given election can usually be done efficiently,
even with arbitrarily-many voters.
Moreover, we report on certain computational simulations. 
\end{abstract}

\ifarxiv
\else
\fi

\ifarxiv
\else
\end{frontmatter}
\fi


\section{Introduction}

The introduction of online social networks to modern politics has
thoroughly changed how political campaigns are run, as currently it is
practical to influence selected individuals or groups of individuals
on a scale not possible before.  Political campaigns now routinely use
these networks to attempt to sway elections in their favor, for
instance, by targeting segments of voters with fake
news~\cite{fakenews,fakenews2}, by organizing fund-raising activities,
and by running vote suppression campaigns~\cite{politics-socialmedia}.
Indeed, the use of social media in election campaigns is so ubiquitous
that there already are hundreds of studies regarding their use (as a
piece of evidence in this respect, we point the reader to Jungherr's
overview of over a hundred Twitter-focused studies~\cite{twitter}).
To be efficient, campaigners would like to factor-in the nuances of
how each voter behaves and how beliefs diffuse in the underlying
social graph.  Yet, doing so is challenging for at least two reasons:
On the one hand, it is difficult---albeit not impossible---to obtain
such a fine-grained understanding of the social-media users and to
prepare the right content for them. On the other hand, as shown by
Wilder and Vorobeychik~\cite{wilder2017controlling} and by Bredereck
and Elkind~\cite{bredereck2017manipulating}, finding optimal
strategies to affect the election results is computationally
challenging even for the most basic voting rules.  To circumvent their
intractability results, Wilder and
Vorobeychik~\cite{wilder2017controlling} designed appropriate
approximation algorithms, whereas Bredereck and
Elkind~\cite{bredereck2017manipulating} considered very restricted
types of social networks. Our goal also is to address the
computational difficulty of manipulating elections through targeting
particular groups of voters, but we take a very different approach.
First, instead of considering networks of individuals, we assume that
the network is over clusters of like-minded voters. Second, we provide
fixed-parameter tractable (FPT) algorithms parameterized by the number
of candidates and the number of these clusters. As a consequence, the running times of our
algorithms scale smoothly with the precision to which
we model the social network.


Specifically, in our model, an external agent with limited funds
observes an election and wants to ensure that a certain candidate
wins. To this end, he or she can alter the preferences of some voters, e.g., by bribing them or by convincing them through some sort
of a targeted campaign (we assume the ordinal election model, where
each voter ranks the candidates from the most to the least appealing
one, and we model campaigning actions---at least in our basic model---as shift
bribery~\cite{elkind2009swap,elkind2010shift-bribery,briberysurvey}). Then,
voters' opinions diffuse through a social network and, eventually, the
election result is established.

The crucial part of our model is that
the social network is not over individual voters, but over clusters of
voters. Each cluster might correspond to a group of voters who share the
same preferences and, possibly, some other features. For example, for
each given preference order we may have three clusters, containing the
voters who share this preference order and are, respectively, young,
middle-aged, or senior. The clusters are connected in the network and
only voters in connected clusters may influence each other (the edges
in our network correspond to the fact that voters from various
clusters interact with each other; e.g., like-minded voters of the
same age may visit the same blogs and read each other's opinions
there). We refer to the networks of our type as \emph{society graphs} (we follow the model of Knop et al.~\shortcite{bestpaperever}).
We consider a diffusion model where voters change
their mind based on the most popular preferences of the voters to
which they are connected, as well as its various generalizations
(indeed, our FPT algorithms can work with a very broad family of
diffusion models that can be expressed via linear programs of a
certain kind).  

Our main contributions are as follows:
\begin{enumerate}
\item We provide the model of elections over society graphs,
  parametrized by both the voting rule and the diffusion
  process. Throughout most of the paper we focus on a simple variant
  of the our model, but we also discuss a very general
  notion of ILP-expressible diffusion processes.

\item We provide an FPT algorithm for our problem parametrized by
  the number of candidates. We present the algorithm for one of the
  simplest variants of our model, with the basic diffusion
  process and voters clustered by their preference order only, but we
  also argue how it can be extended in numerous ways (we focus on the
  simplest setting mostly for the sake of clarity).  In particular, as
  the algorithm is based on expressing the problem via integer linear
  programming, it works for all ILP-expressible voting rules and all
  ILP-expressible diffusion processes.  Our algorithm runs in
  exponential time with respect to the number of candidates, but in
  logarithmic time with respect to the number of voters, so it is 
  particularly well-suited for the case of political elections.

\item We test our algorithm experimentally. It turns out that for
  modern ILP solvers it is still quite challenging
  (although usable for at least some realistic settings). Thus we provide
  two heuristics algorithms that are much faster, but whose outputs are
  not guaranteed to be optimal.
\end{enumerate}

We note that our model and our society graphs in particular form a
generalization of the standard model, where each voter is a node in
the social network. Indeed, we can simply assume that each cluster
contains a single individual. The advantage of our approach is that we
can seamlessly move between more and less fine-grained views of the
society and its interactions. As our model can capture many natural
social behaviors within the clusters, a plethora of bribery actions,
and various diffusion processes, we believe it to be quite powerful; indeed, we demonstrate its power by highlighting many modeling possibilities within it.

Yet, our model---or, more specifically, our approach---has one
drawback: Our diffusion models are completely deterministic in that
the effect of the diffusion is always exectly the same for given
initial conditions. If one prefers to have a model with some sort of
stochastic behavior,
then there is a natural work-around for this issue. For example,
Wilder and Vorobeychik~\cite{wilder2017controlling} argue that by
sampling several possible ``diffusion scenarios'' (they use the
Independent Cascade model, so in their case it means sampling which
edges of the network indeed propagate the influence) one can get a
very good approximation of the behavior of a stochastic diffusion
process. Since the same approach can be used in our case, we find it
sufficient to consider deterministic diffusion models
(and, indeed, deterministic models are commonly studied in
the literature~\cite{peleg-deterministic,contagion,brill2016pairwise,auletta2015minority}).

The paper is organized as follows. First, in
Section~\ref{section:related} we discuss related work regarding
algorithmic results in the areas of social networks and election
bribery. Then, we describe our model in
Section~\ref{section:preliminaries} and analyze the properties of our
basic diffusion process in Section~\ref{section:convergence}. In
Section~\ref{section:complexity}, we provide our main algorithmic
results, while in Section~\ref{section:model extensions} we discuss
possible extensions of the basic model and how our algorithm can be
adapted for them. We conclude by discussing experimental results in
Section~\ref{section:experiments} and by discussing possible future
work in Section~\ref{section:outlook}.


\section{Related Work}\label{section:related}

We study the possibility of manipulating election outcomes under the
assumption that the voters' views propagate (or, diffuse) throughout
an underlying social network. Below we present related work regarding
social networks, bribery in elections, and a few results related by technique.

\subsection{Diffusion in Social Networks}

The two papers most closely related to
our work, one due to Wilder and
Vorobeychik~\cite{wilder2017controlling} and one due to Bredereck and
Elkind~\cite{bredereck2017manipulating}, study very similar issues,
but differ in several important modeling choices.  Foremost, we
consider the model of society graphs, where the nodes of the social
network represent clusters of voters, whereas both Wilder and
Vorobeychik and Bredereck and Elkind consider the more typical model
where each vertex represents a single individual. Furthermore, both Wilder
and Vorobeychik and Bredereck and Elkind focus on the simple Plurality
voting rule, whereas we consider a wide spectrum of rules (indeed, all
rules that are ILP-expressible, in essence including all commonly studied rules).  Regarding the diffusion process, Wilder and
Vorobeychik consider the Independent Cascade model (ICM), and Bredereck
and Elkind consider the Linear Threshold model (LTM).

It is worth exploring the differences between ICM, LTM, and our diffusion model.
In ICM and LTM, the underlying idea is that a message is spreading from multiple seed nodes, and this message has a certain ``directionality''---its intent is to activate nodes (which can be understood as, e.g., convincing
them of some fact).
So, these two models can be seen as modeling an intentional act of  campaigning or influencing.
In contrast, our diffusion model seeks to describe what happens as a result of mutual influence (or peer pressure) between voters.
Indeed, in our model the voters observe the other ones---typically, those close to them in the network---and change their preferences accordingly. As a consequence, in our model a voter may change his preference order arbitrarily, if faced with appropriate pressure. In contrast, classic variants of ICM and LTM 
only capture binary influence---either a node is influenced or it is not (which is why influence maximization can be viewed in terms of manipulating a Plurality election over two candidates). Yet, we mention that recently 
Cor{\`o} et al.~\cite{coro2019exploiting} proposed a variant of LTM, called Linear Threshold Ranking, LTR, where the level of influence is quantified in a more fine-grained way (with stronger infuence, a designated candidate may be shifted up by more positions; interestingly, this is quite related to shift
bribery, which we also consider).


Another difference between our diffusion model and ICM, LTM, and LTR is determinism: our model is deterministic while the other models are stochastic.
There are a few reasonable ways to turn our model into a stochastic one, and we believe that our algorithmic results may carry over using the 
trick of considering several ``diffusion scenarios'' (in the fashion of
Wilder and Vorobeychik), but such discussion is beyond the scope of this paper.
For a general, broad
discussion of diffusion processes, we point the reader to the book of
Shakarian et al.~\shortcite{shakarian2015diffusion}).

The main technical difference between our work and the papers of Wilder and Vorobeychik, and Bredereck and Elkind, is in
how we deal with computational intractability: Wilder and Vorobeychik
provide approximation algorithms (but also MILP formulations),
Bredereck and Elkind consider very restricted classes of social
networks, and we design fixed-parameter tractable algorithms,
parametrized by the number of candidates and the number of voter
clusters. Our main algorithm proceeds by forming and solving an ILP
instance and, in this sense, it is similar to the MILP approach of
Wilder and Vorobeychik. Our ILP, however is very different and, in particular, incorporates very different tricks.

Another paper that is very closely related to our work is that of
Silva~\cite{silva2016opinion}, where the author studies a similar
bribery model, but for cardinal preferences expressed as numbers
between $0$ and $1$ (e.g., corresponding to the level of support for
the current government). In his model, these values can be modified
through bribery actions, after which a certain dynamic process (i.e.,
a form of diffusion) propagates them over the network. Silva focuses
on a Twitter-like network, where edges are directed and influence can
only go in one way, whereas we---as well as Wilder and
Vorobeychik~\cite{wilder2017controlling} and Elkind and
Bredereck~\cite{bredereck2017manipulating}---consider more
Facebook-like networks, where the interactions are two-way.
(Note that we nevertheless later show how to also handle directed edges in our model; see Section~\ref{subsec:extensions}.)
Similarly to us, Silva uses integer linear programming to compute the solutions
to his problems.

More broadly, our work is closely related to the stream of papers on
the interface between social choice and social networks, which
includes the problems of recovering ground
truth~\cite{conitzer-network-mle,procaccia2015ranked}, the problems of
iterative voting~\cite{manipulation-social-network,sina2015adapting}, various issues
related to liquid democracy~\cite{liquid-demorcacy}, certain forms of
multiwinner voting~\cite{talmon2017structured}, and many others (for a
more detailed discussion, see the overview of Grandi et
al.~\cite{chaptergrandi}).
In particular, our work is quite closely related to that of Brill et
al.~\cite{brill2016pairwise}, who study the diffusion of ordinal
preferences through a social network. In their network,
each vertex (i.e., each individual) swaps two candidates if the
majority of its neighbors ranks them differently than him. This is
quite similar to our basic model, where connected clusters of voters
have identical preference orders, up to a single swap.  The two main
differences are that, on the one hand, our network is more restrcited
and, on the other hand, we allow for the initial modification of some
of the preference orders, whereas Brill et
al.~\cite{brill2016pairwise} focus only on the diffusion.
Similarly, our work is related to that of Botan et
al.~\shortcite{botan2017propositionwise}, where the authors consider a
diffusion process of preferences expressed as Boolean propositions.
In a similar vein, Christoff and Grossi~\cite{christoff2017stability}
study the convergence of binary opinions over a network.

Our work is also naturally connected to the stream of work on
influence maximization in social networks (see, e.g., the papers of
Kempe et al.~\cite{kdd03,kempe05influence} and Chen et
al.~\cite{chen2009efficient} as well as many follow-up ones). In these
works, the goal is to choose a set of nodes in a social network so
that if we pass some information to them and wait for the diffusion
process to converge, then as many nodes as possible will have received
our information (the work of Bredereck and
Elkind~\cite{bredereck2017manipulating} can be understood in these
terms as well; alternatively, influence maximization can also be seen
in terms of manipulating Plurality voting).

\subsection{Bribery in Elections}

Stepping away from social networks, our work belongs to the broad
stream of papers on the complexity of manipulating elections.  For a
general overview of this topic, we point the reader to the surveys of
Conitzer and Walsh~\cite{manipsurvey} and Faliszewski and
Rothe~\cite{briberysurvey}; here we will discuss the few most related
papers.

We model campaigning actions via the shift bribery problem, which
itself is a special case of the swap bribery problem; both introduced
by Elkind et al.~\cite{elkind2009swap,elkind2010shift-bribery}.
Briefly put, in the shift bribery problem we are given an election,
where each voter ranks all the candidates from the most to the least
appealing one, and, upon paying a required price, we can ask some of
the voters to shift a certain candidate $p$ higher. Our goal is to
ensure that $p$ wins, but without exceeding a given budget. Swap
bribery generalizes shift bribery by allowing swaps of any adjacent
candidates, and not just $p$, with those who preceed him or her.
Both shift bribery and swap bribery are $\np$-hard for many natural
voting rules (including, e.g., Borda, Copeland, Maximin, and various
elimination-based rules~\cite{maushagen2018iterative-shift-bribery}),
but shift bribery is generally easier to deal with. For example,
Elkind et al.~\cite{elkind2009swap,elkind2010shift-bribery} provided
approximation algorithms for shift bribery under several voting rules
(recently strengthened by Faliszewski et
al.~\cite{shiftbribery-ptas}), Bredereck et
al.~\cite{bredereck2016complexity} and Zhou and
Guo~\cite{zho-guo:c:param-iter-shift-bribery} provided several FPT
algorithms, and Elkind et al.~\cite{elkind2020spsc-shift-bribery} gave
polynomial-time algorithms for several structured preference
domains. For the case of swap bribery, Dorn and
Schlotter~\cite{DornS:2012} provided a careful analysis for the case
of approval voting, whereas Knop et al.~\cite{knop-fpt-bribery} gave a
general FPT algorithm parameterized by the number of candidates.
Both problems were also studied in the destructive setting, where the
goal is to prevent a given candidate from being a
winner~\cite{kacz-dsb,shir-dest-swap}. Interestingly, in this case the
problem is often efficiently solvable. Nonetheless, we do not consider
the destructive setting in our work (we do not expect the tractability
results to carry over to our model due to the computational
complexity implied by the diffusion process).

We stress that our choice of shift bribery as a model for campaigning
bears only limited significance regarding the complexity of our
problem. Indeed, we could have used full-fledged swap bribery or the
classic bribery problem of Faliszewski at
al.~\cite{faliszewski2009bribery} or any other problem from the
bribery family~\cite{briberysurvey} and---as long as it were
expressible within an ILP program---the complexity of our algorithm
would stay intact.

\subsection{Other Related Work}

Finally, we mention that the idea of using society graphs and clusters
of voters of a given type was heavily inspired by the work of Knop et
al.~\cite{bestpaperever}, who studied a very general form of
manipulating elections. In their work, an election is represented as a
society vector $(s_1,\ldots,s_n)$, where each $s_i$ specifies how many
voters of type~$i$ there are (e.g., how many voters have the $i$-th
possible preference order). The goal is to find a minimum-cost
transformation of a society vector to one satisfying a given
condition, provided a certain cost measure for transforming the
society. Our work extends the approach of Knop et
al.~\cite{bestpaperever} to include the social network in the
election.
The idea of types was also considered, e.g., by Izsak et
al.~\shortcite{interclass}, but for the case of candidates.

\section{Formal Model and Combinatorial Problem}\label{section:preliminaries}

In this section we present a very basic variant of our model, where
voter types correspond to preference orders, edges exist between two
orders that can be obtained by a single swap of adjacent candidates,
the diffusion process is done in a simple, particular way, and the
bribery actions are limited.  Later, in Section~\ref{section:model
  extensions}, we discuss various generalizations of our approach, by
considering arbitrary voter types, arbitrary bribery actions, and
generalized diffusion processes.  Yet, the basic model will allow us
to develop intuitions, prove strong hardness results, and present our tractability results clearly.  For $n \in \N$, by $[n]$ we mean
the set $\{1, \ldots, n\}$.

\subsection{Elections and Voting Rules}

We consider ordinal elections held with $n$ voters, expressing
preferences over $m$ candidates $C = \{c_1, \ldots, c_m\}$, where the
\emph{preference order} of a voter is a linear order over $C$.  A
voting rule~$\calR$ is a function taking an election as input and
returning a set of tied winners.  A candidate winning under $\calR$
for a given election is called an $\calR$-winner of the election.
As an example, under the \emph{Plurality} rule the candidates ranked
first most frequently win, and under the \emph{Borda} rule, for each
position $i$ each voter gives $m - i$ points to the candidate ranked
there; the candidates with the highest total number of points win.

\subsection{Voter Types, Societies, and Society Graphs}

For the time being, we let the preference order of
a voter be her
 \emph{type}. 
Thus, there are $\tau \leq m!$ types present in a given election and
we order them arbitrarily so that we can speak of ``the $j$-th type'' for a given $j \in [\tau]$.
By the \emph{weight} of voters with type $j$, denoted either $w_j$ or $w(j)$, as is more convenient, we mean the number of voters of type $j$.
%
Sometimes we represent an election as a vector $\mathbf{w} \in \mathbb{N}^\tau$,
whose $j$-th entry represents the weight of type $j$. We refer to such vectors
as \emph{societies}.
(in this we follow the model of Knop et al.~\shortcite{bestpaperever}).

As we are interested in diffusion processes operating on the voter types,
we associate a given election with a vertex-weighted graph $G = (V, \mathbf{w}, E)$,
termed the \emph{society graph}.
The society graph contains $\tau$ vertices, where $\tau$ is the number of types in the election
(specifically, $V = \{v_1, \ldots, v_\tau\}$,
where vertex $v_j$ corresponds to voter type~$j$,
and its weight $w_j$ is equal to the number of voters of that type in the given election).
There is an edge between vertices $v_j$ and $v_{j'}$ if the preference
orders corresponding to types $j$ and $j'$ differ by the ordering of a single pair of adjacent
candidates (in other words, if it is possible to transform one into the other with a single
swap of two consecutive candidates).
We show an example of a society graph in Figure~\ref{figure:society
  graph}.\footnote{Graphs of this form are quite popular in the study
  of permutations. Regarding their use in the context of elections, we
  point out, e.g., to the work of Puppe and
  Slinko~\cite{puppe2017condorcet-domains}. Recall that later we will also consider other graphs.}

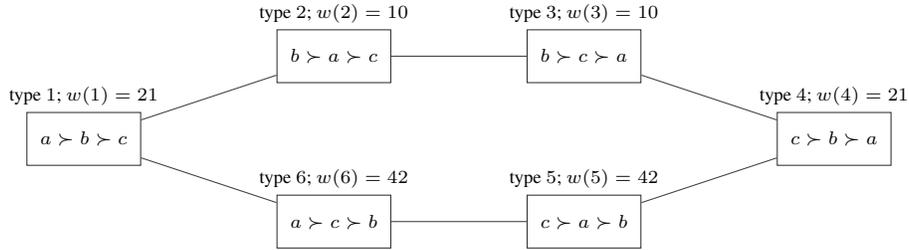
\begin{figure}[t]
    \scriptsize
    \center
    \begin{tikzpicture}[draw=black!75, xscale=1.33, yscale=2.2,-]
      \tikzstyle{original}=[rectangle,draw=black!80,minimum size=20pt,inner sep=5pt]

      \foreach [count=\i] \pos / \text / \above / \style in {
        {(0,0.5)}/$a \pref b \pref c$/type $1$; $w(1) = 21$/original,
        {(2.5,1)}/$b \pref a \pref c$/type $2$; $w(2) = 10$/original,
        {(5,1)}/$b \pref c \pref a$/type $3$; $w(3) = 10$/original,
        {(7.5,0.5)}/$c \pref b \pref a$/type $4$; $w(4) = 21$/original,
        {(5,0)}/$c \pref a \pref b$/type $5$; $w(5) = 42$/original,
        {(2.5,0)}/$a \pref c \pref b$/type $6$; $w(6) = 42$/original}
      {
        \node[\style,label=above:{\above}] (V\i) at \pos {\text};
      }
      
      \foreach \i / \j in {1/2,2/3,3/4,4/5,5/6,6/1} {
        \path[] (V\i) edge (V\j);
      }

    \end{tikzpicture}
    \caption{A society graph with three candidates and six types (corresponding to the six possible preference orders on those three candidates). In this graph there are, e.g., $42$ voters of type $6$, each with preference order $a \pref c \pref b$; this graph corresponds to a society $\mathbf{w} = [21, 10, 10, 21, 42, 42]$.}\label{figure:society graph}
  \end{figure}

\subsection{Diffusion of Preferences}

Given a society graph (which encodes a given election), we consider
two variants of the diffusion process, namely asynchronous
and synchronous.  In the asynchronous variant, in each step of the
process some vertex~$v$ of the society graph $G$ is picked and, then,
the following occurs (we do not specify which vertex is selected and,
as we will see in Example~\ref{example:convergence order} below,
different orders of selecting the vertices may lead to different
outcomes of the process).  We consider the closed neighborhood $N[v]$
of $v$ in $G$ and check whether there is a neighbor $x$ of $v$ for
which $w_x > \nicefrac{1}{2} \sum_{u \in N[v]} w_u$; that is, a
neighbor whose weight exceeds the sum of the weights of all other
vertices in the closed neighborhood of $v$.
If such a neighbor $x$ exists, then we add the current weight $w_v$ of
$v$ to that of $x$ and change the weight $w_v$ to be $0$. Intuitively,
the voters of type represented at $v$ look at all the voters with
similar or identical preferences and if there is a majority support
among these voters for some preference order, then they switch to it.
In the synchronous variant we proceed  in the same way, but simultanously for all vertices.
The diffusion process halts whenever it stabilizes.

\begin{example}\label{example:convergence order}
  Consider the society graph depicted in Figure~\ref{figure:society
    graph} and asynchronous diffusion. Assume that we first choose
  type $3$. As type $3$ has as neighbors types $2$ and $4$, together
  there are $41$ voters of these types, and $21$ of them have
  preference order $c \pref b \pref a$. So, the $10$ voters with type $3$
  move to have type~$4$. If we then select type~$4$, type~$6$, and
  then type~$2$, then the diffusion converges with $115$ voters of
  type~$5$ (with preference order $c \pref a \pref b$) and $31$ voters
  of type~$1$ (with preference order $a \pref b \pref c$); thus,
  Plurality selects $c$.  However, if we select first type~$2$, then
  $1$, then $5$, and then $3$, then we reach convergence with $115$
  voters of type $6$ (with preference order $a \pref c \pref b$)
  and $31$ voters of type~$5$ (with preference order
  $c \pref b \pref a$); thus, Plurality selects $a$.  This shows that the asynchronous diffusion process can lead to different outcomes, depending on the order in which vertices are considered.

  Let us now consider the same society graph and synchronous
  diffusion. After the first round, we have $10$ voters of type~$1$
  (voters of type~$2$ moved to have type~$1$, whereas original
  type~$1$ voters moved to have type~$6$), no voters of types~$2$
  and~$3$, $10$ voters of type~$4$, $63$ voters of type~$5$, and $63$
  voters of type~$6$. After the next round there are $73$ voters of
  type~$5$ and $73$ voters of type~$6$. No further changes are
  possible and the process converges; Plurality selects $a$ and $c$ as
  two tied winners.
\end{example}

\subsection{Bribery in Society Graphs}

Besides issues related to the diffusion of preferences,
we are mainly interested in understanding the possibility of
manipulating election outcomes.  Thus we assume that there is an
external briber who has some budget and, using this budget, can affect
the original preference orders of some voters (i.e., the preference
orders they have prior to the diffusion).
Specifically, in a single bribery action the briber chooses a single
voter and, at unit cost, shifts the briber's preferred candidate $p$
up by one position in this voter's preference order (in effect,
changing this voter's type; see the work of Elkind et
al.~\cite{elkind2009swap,elkind2010shift-bribery} and Bredereck et
al.~\cite{bredereck2016prices} for a detailed discussion of shift
bribery and its various cost models).
The briber performs as many bribery actions as he wants, up to the
budget limit, and then the diffusion process takes place.  The goal of
the briber is to have his preferred candidate $p$ win the resulting
election (under a given, predetermined voting rule).  Formally, we are
interested in the following general problem.


\probDef
  {$\calR$-Bribery in society graphs ($\calR$-BSG)}
  {A society graph $G$ (given directly as a graph), a preferred
    candidate $p$, and a budget $b$.}
  {Are there at most $b$ (unit-cost, shift-) bribery actions, such
    that after performing them on $G$ and then running the diffusion
    process, $p$ is an $\calR$-winner of the resulting election?}

  
Corresponding to the synchronous and asynchronous diffusion processes,
we consider both \emph{sync-$\calR$-BSG} and \emph{async-$\calR$-BSG}
problems.  For the asynchronous diffusion, we further consider the
\emph{optimistic} and \emph{pessimistic} variants of the problem. In
the former, we ask whether the briber's preferred candidate wins for
\emph{some} order of the diffusion steps. In the latter, we require that $p$
wins for \emph{every} order of diffusion steps that leads to convergence.



\begin{remark}
  The input to $\calR$-BSG is a labeled graph with weighted vertices,
  a preferred candidate $p$, and a budget $b$.  Thus the size of the
  input encoding is linear in the number of voter types and only
  logarithmic in the number of voters.
\end{remark}

\section{Convergence and Diffusion Order}\label{section:convergence}

Before we tackle the $\calR$-BSG problem, we first show that our
diffusion processes
always converge, but the complexity of deciding if a particular
candidate may become a winner due to the diffusion may be $\np$-hard.

For the synchronous case, convergence follows by arguing that in each
diffusion step at least one vertex loses its weight completely, and a
vertex of weight zero never increases its weight.  In consequence, we have that the
number of synchronous diffusion steps is bounded by the number of
voter types.
The asynchronous case is even simpler, but requires appropriate
terminology: If a diffusion step
does not change the society graph (e.g., due to the choice of the
vertex) then we call it \emph{redundant}.  We say that a sequence of
non-redundant diffusion steps is \emph{irredundant}.  A \emph{maximal
  irredundant sequence} is an irredundant sequence after executing
which all remaining steps are redundant.

\begin{proposition}\label{proposition:convergence}
  For each society graph $G$, the synchronous diffusion process
  converges in at most $\tau$ steps.  The asynchronous diffusion
  process converges if the sequence of diffusion steps contains a
  maximal irredundant sequence as a subsequence.  The length of a
  maximal irredundant sequence is bounded by~$\tau$.
\end{proposition}

\begin{proof}
  We consider the asynchronous case first. Consider some diffusion
  step. At this point either no vertex changes any further, or at
  least one vertex, if chosen for the next diffusion step, would have
  its weight reduced to zero.  Since no weight-zero vertex can ever
  increase its weight (by the definition of the diffusion step), it
  follows that every irredundant sequence consists of at most $\tau$
  steps.  By definition, if a sequence contains a maximal irredundant
  subsequence, it produces the same graph as this subsequence.

  For the synchronous case, it suffices to show that after every
  diffusion step (prior to convergence), the number of vertices with
  non-zero weight decreases.  Consider a diffusion step before
  convergence.  There is some vertex $v$, which is to be assimilated
  into one of its neighbors, $u$.  If no other neighbor of $v$ is to
  be assimilated by $v$ in this step, then we are done: The number of
  vertices with non-zero weight will decrease by at least one after
  this diffusion step.  Perhaps, however, there is some neighbor $v'$
  of $v$ that is to be assimilated by $v$ in the current step. It must
  be that $v' \neq u$, as we require a strict majority for a vertex to
  be assimilated by one of its neighbors.
%
  If no neighbor of $v'$ is assimilated by $v'$, then we are done (by
  the same token as before, we see that the number of weight-zero
  vertices will increase).  Otherwise, there is some neighbor $v''$ of
  $v'$ which is to be assimilated by $v'$.  By following this logic
  exhaustively, either we reach a vertex whose weight is to decrease
  to zero, or some vertex repeats. However, the latter is impossible
  as, by definition of the diffusion process, the weights of the
  vertices that we encounter form a decreasing sequence. Thus the
  number of weight-zero vertices increases after each diffusion step
  and the claim follows.
%
%
\end{proof}





For synchronous diffusion, the final society graph is defined uniquely
and so is the outcome of the election (for a given voting rule).  This
is not the case for asynchronous diffusion. Indeed, in
Example~\ref{example:convergence order} we have seen that two
different diffusion orders may lead to two different society graphs.
In the next theorem we show a stronger statement, namely that the problem of
deciding whether a given candidate may be a Plurality winner after
asynchronous diffusion is $\np$-hard.\footnote{The proof of
  Theorem~\ref{theorem:diffusion matters} presented below is different
  than the one included in the conference version of this paper. The
  latter had a technical flaw due to which the society graph used in
  the construction was not implementable. The proof presented below is
  based on a different construction and fixes this issue.  As an added
  benefit, it uses weights whose values are polynomially bounded with
  respect to the size of the input instance (i.e., it shows that the
  problem is strongly NP-hard, whereas the previous proof was showing
  weak NP-hardness).}

\begin{theorem}\label{theorem:diffusion matters}
  Given a society graph $G$ and a preferred candidate $p$, deciding
  whether there is an order of asynchronous diffusion steps that
  results in $p$ being a Plurality winner in the converged election is
  $\np$-hard.
\end{theorem}
\begin{proof}
  We provide a reduction from the $\np$-complete problem \textsc{Cubic Vertex
  Cover}~\cite{gar-joh-sto:j:simplified-graph-problems}. In this
  problem we are given a graph $G$, where each vertex has degree at
  most three, and an integer $k$ (sometimes authors assume that the degree of the vertices in the graph is exactly three, but it is not needed in our case). We ask whether it is possible to select
  at most $k$ vertices so that each edge touches at least one of the
  selected vertices.

  Consider an instance of \textsc{Cubic Vertex Cover} that consists of
  a graph $G$ and an integer $k$.
  We write $V(G) = \{v_1, \ldots, v_n\}$ to denote the set of
  $G$'s vertices and $E(G) = \{e_1, \ldots, e_m\}$ to denote the set
  of its edges. We build the following election (and the associated
  society graph). We let the candidate set be
  $C = \{c,d,e,p\} \cup A \cup B$, where $A = \{a_1, \ldots, a_n\}$
  and $B = \{b_1, \ldots, b_n\}$. The role of the candidates in sets
  $A$ and $B$ is to encode subsets of $V(G)$. We sometimes use the following
  convention to describe the nodes of the society graph: Given three
  candidates $x,y,z \in \{d,e,p\}$, some subset $S$ of $V(G)$, and an
  integer $w$, by $w / xyzS$ we mean a node of the society graph with
  weight $w$ and preference order obtained from:
  \[
    x \pref y \pref z \pref a_1 \pref b_1 \pref a_2 \pref b_2 \pref \cdots \pref a_n \pref b_n \pref c
  \]
  by swapping for each $v_i \in S$ candidate $a_i$ with candidate
  $b_i$ (intuitively, candidates from the set $A \cup B$ create a
  signature that puts our preference order at sufficient swap
  distance from other votes that rank candidates $x,y,z$ in the same
  way).

  To form the society graph, let $T$ be some large positive integer (to
  be specified later), let $X = 10$, let $Y = 100$, and let $Z = 1000$. 
  The society graph consists of two parts. The first part contains three isolated
  vertices:
  \begin{enumerate}
  \item A vertex with candidate $c$ ranked first and with weight $T$.
  \item A vertex with candidate $p$ ranked first and with weight $T - (2X+9)|E(G)|$.
  \item A vertex with candidate $d$ ranked first and with weight $T - (Y+Z)|V(G)|-kX$.
  \end{enumerate}
  By analyzing the second part of the society graph, it will become
  clear that it is easy to make sure that these vertices are indeed isolated (for the first one, it
  suffices to rank $c$ on the first place and to rank all the other
  candidates arbitrarily; for the latter two it suffices to rank the
  requested candidate on top, followed first by $c$ and then all the other
  candidates arbitrarily).

  The second part of the society graph encodes the graph $G$. We build
  it as follows:
  \begin{enumerate}
  \item For each vertex $v_i \in V(G)$, we form three nodes, $v^{(1)}_i$, $v^{(2)}_i$, and $v^{(3)}_i$,
    specified as follows (recall our convention for describing the
    nodes):
    \begin{align*}
      v^{(3)}_i \colon& Z / dpe\{v_i\}, &
      v^{(2)}_i \colon& Y / dep\{v_i\}, &
      v^{(1)}_i \colon& X / edp\{v_i\}.    
    \end{align*}
    Note that these vertices form a path
    $v^{(3)}_i$--$v^{(2)}_i$--$v^{(1)}_i$ and that there are no edges
    between society-graph nodes associated with different vertices
    $v_i, v_j \in V(G)$.

  \item For each edge $e_t = \{v_i, v_j\} \in E(G)$ we form three
    nodes:
    \begin{align*}
      e^{(1)}_t \colon& 1 / edp\{v_i,v_j\}, &
      e^{(2)}_t \colon& X+3 / epd\{v_i,v_j\}, &
      e^{(3)}_t \colon& X+5 / ped\{v_i, v_j\}.    
    \end{align*}
    These vertices form a path $e^{(1)}_t$--$e^{(2)}_t$--$e^{(3)}_t$
    and there are no edges between society-graph nodes associated with
    different edges from $E(G)$. However, there is an edge between
    $e^{(1)}_t$ and $v^{(1)}_i$, and between $e^{(1)}_t$ and
    $v^{(1)}_j$.
  \end{enumerate}
  The society graph does not contain any edges aside from those
  explicitly mentioned in the construction above.  We set $T$ to be the
  square of the sum of the weights of the nodes from the second part of
  the society graph (intuitively, $T$ is simply a large number).
  This completes the construction. See Figure~\ref{figure:npone} for an example with a small input graph.
  
  \begin{figure}
	\centering
	  \includegraphics[width=12.5cm]{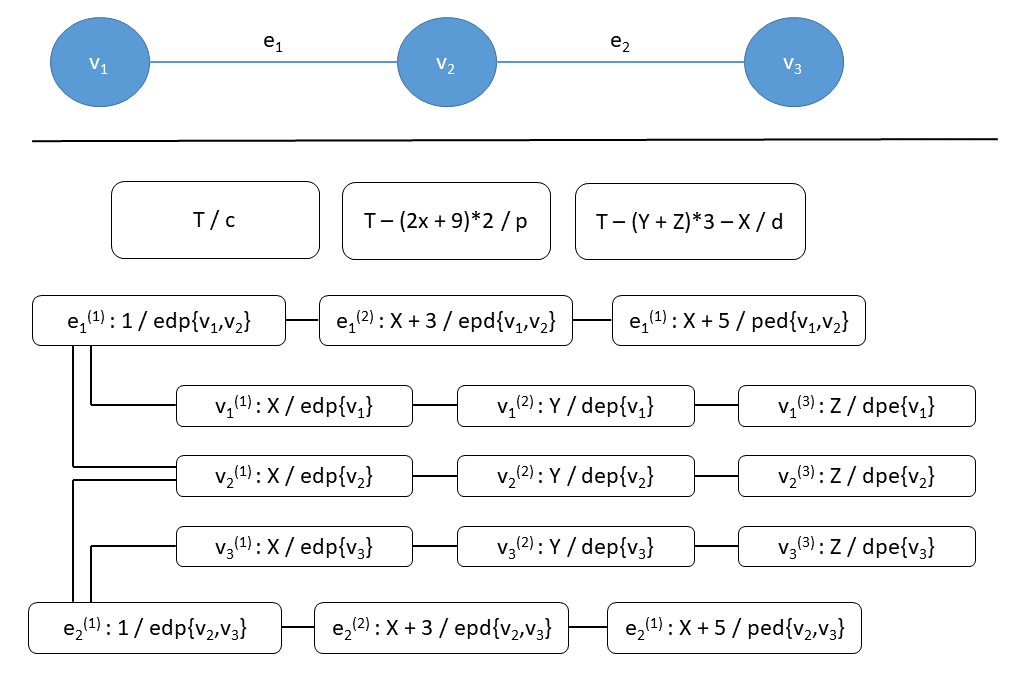} 
	\caption{An example for the reduction described in the proof of Theorem~\ref{theorem:diffusion matters}. On top, we show a graph containing $3$ vertices and $2$ edges. At the bottom, we show the constructed society graph.}
	\label{figure:npone}
\end{figure}
  
  Prior to the diffusion, the
  candidates from sets $A$ and $B$ have score $0$ (and their score
  never increases), whereas the remaining candidates have the
  following scores:
  \begin{enumerate}
  \item candidate $c$ has score $T$ (this is currently the highest
    score and it cannot change; any candidate that becomes a winner
    after the diffusion has to reach at least score $T$),
  \item candidate $d$ has score $T-kX$,
  \item candidate $p$ has score $T-(X+4)|E(G)|$, and
  \item candidate $e$ has score much below $T$ (and this score cannot
    reach $T$ even due to the diffusion).
  \end{enumerate}

  Let us now argue that if there is a set $S$ of at most $k$ vertices
  from $V(G)$ such that each edge in $E(G)$ touches at least one
  vertex from $S$, then there is a diffusion order that leads to $p$
  being a Plurality winner of the given election. Indeed, one such
  diffusion order proceeds as follows:
  \begin{enumerate}
  
  \item
  First, for each $v_i \in S$,
  society-graph node $v^{(1)}_i$ is assimilated into $v^{(2)}_i$.
  
  \item
  Then, for each $j \in [n]$, society-graph node $v^{(2)}_j$ is
  assimilated into $v^{(3)}_j$.
  Note that, after this happens, there are no more nodes
  of the form $v^{(\ell)}_i$ that can be assimilated into any other
  node.

  \item
  Next, for each $t \in [m]$, node $e^{(1)}_t$ is assimilated into
  $e^{(2)}_t$ (this is possible because, as $S$ is a vertex cover, now
  each $e^{(1)}_t$ is connected to at most one node of the form
  $v^{(1)}_i$, with weight $X$, and to exactly one node $e^{(2)}_t$,
  with weight $X+3$; as the total weight that $e^{(1)}_t$ sees is
  $2X+4$, its weight can move to $e^{(2)}_t$).
  
  \item
  Finally, for each
  $t \in [m]$, node $e^{(2)}$, currently with weight $X+4$, is
  assimilated into $e^{(3)}_t$, who (prior to this) has weight $X+5$.  
  Note that, after this happens,
  no further diffusion steps are possible. Moreover, $p$ and $c$ have
  score $T$, $d$ has score at most $T$, and $e$ has score below $T$.
  Thus $p$ is a winner.

\end{enumerate}

  Let us now consider the other direction, i.e., the case that $p$ can
  become a winner of our election for some diffusion order. For this
  to happen, $p$ has to reach score at least $T$ and this is possible
  only if for each $t \in [m]$, the weight from nodes $e^{(1)}_t$ and
  $e^{(2)}_t$ is assimilated into the node $e^{(3)}_t$, which ranks
  $p$ on top (note that no other nodes can transfer their weight to
  nodes that rank $p$ on top). However, for a given $t \in [m]$, the
  weight from $e^{(1)}_t$ can reach $e^{(3)}_t$ only by first being
  assimilated by $e^{(2)}_t$.  Yet, prior to the diffusion, $e^{(1)}$
  sees total weight equal to $3X+4$ and cannot be immediately
  assimilated by any of its neighbors ($e^{(1)}_t$ sees its own weight
  of $1$, weight $X+3$ of $e^{(2)}_t$ and weight $X+X$ of the two
  nodes $v^{(1)}_i$ and $v^{(1)}_j$, such that edge $e_t$ connects
  $v_i$ and $v_j$).  Thus, for each $e_t = \{v_i,v_j\} \in E(G)$, at
  least one of $v^{(1)}_i$, $v^{(1)}_j$ has to be assimilated into,
  respectively, $v^{(2)}_i$ or $v^{(2)}_j$. However, each such
  assimilation increases the score of $d$ by $X$. If this happens for
  more than $k$ nodes, then the score of $d$ exceeds $T$ and $p$ does
  not become a winner. In other words, if $p$ becomes a winner, then
  there is a set of $k$ vertices in $V(G)$ that touch all of the edges
  from $E(G)$. This means that there is a vertex cover of size $k$ in
  our input graph.
\end{proof}

We have phrased Theorem~\ref{theorem:diffusion matters} to speak of
the Plurality rule as it is the simplest, yet most widely used voting rule. Nonetheless, similar results hold for many other
voting rules. Indeed, it would be rather remarkable if there were a
natural voting rule for which it did not hold.

\section{Complexity of Manipulating Society Graphs}\label{section:complexity}

In this section we present our main theoretical results. Briefly put,
$\calR$-BSG is intractable (both in the synchrounous and asynchronous
variants) for nearly all natural voting rules, but is fixed-parameter
tractable with respect to the number of candidates.


\subsection{General Intractability of BSG}

As $\calR$-BSG is, in essence, a variant of the Shift Bribery
problem~\cite{elkind2009swap}, it naturally inherits most of its
hardness results.  The difference between $\calR$-BSG and standard
Shift Bribery is that the former involves the diffusion process after
the bribery. In the reduction below, we ``turn off'' this diffusion by
ensuring that every voter from a Shift Bribery instance forms an
isolated vertex in the society graph (as the diffusion anyhow does not
happen in isolated graphs, our proof works for both the synchronous and asynchronous
cases). We show the reduction for the Borda rule, but then we argue
that it is also applicable for many other rules.

\begin{proposition}\label{theorem:borda bsg np hard}
  Borda-BSG is $\np$-hard for both the synchronous and the asynchronous
  cases.
\end{proposition}

\begin{proof}
  An instance of Borda-Shift Bribery with Unit Costs (Borda-SB for
  short) consists of an election with candidate set
  $C = \{c_1, \ldots, c_m\}$, voters $v_1, \ldots, v_n$, a
  distinguished candidate $p \in C$, and budget $b$. We ask if it is
  possible to ensure that~$p$ becomes a Borda winner of this election
  by performing $b$ unit-cost shift bribery actions.  Borda-SB is
  well-known to be $\np$-hard~\cite[Proposition
  3]{bredereck2016prices}.

  Given an instance of Borda-SB, we create an instance of Borda-BSG.
  The idea is to alter the instance so that, even after any set of
  bribery actions, the swap distance between each two voters would be at least two; this will ensure that each voter is a singleton in the society graph
  and will prevent diffusion from happening. As Borda-BSG is, in
  essence, Borda-SB with diffusion, the result will follow.
  To this end, our instance of Borda-BSG is the same as the input
  Borda-SB distance, but with the following two changes:
  \begin{enumerate}
  \item We introduce two sets of additional dummy candidates,
    $D = \{d_1, \ldots, d_n\}$ and $E = \{e_1, \ldots, e_n\}$.
  \item We set the preference
    orders of the voters as follows. For each voter $v_i$, let $\mathrm{pref}(v_i)$ mean
    $v_i$'s original preference order regarding the candidates $c_1, \ldots, c_m$.
    We extend this preference order to be:
    $
      \mathrm{pref}(v_i) \pref d_i \pref D \setminus \{d_i\} \pref e_i
      \pref E \setminus \{e_i\},
    $
    where by $D \setminus \{d_i\}$ and $E \setminus \{e_i\}$ we mean,
    respectively, the preference orders
    $d_1 \pref d_2 \pref \cdots \pref d_n$ and
    $e_1 \pref e_2 \pref \cdots \pref e_n$ with candidates $d_i$ and
    $e_i$ removed.
  \end{enumerate}
  One can verify that in the resulting instance of Borda-BSG each two
  voters have preference orders that are at swap distance at least two
  and, thus, each voter is a singleton in the society graph associated
  with the instance. Moreover, the dummy vertices do not change the winner. As a consequece, the correctness of the reduction
  follows.
\end{proof}

\begin{remark}
The above proof works for every voting rule (i) for which Shift
Bribery with unit costs is $\np$-hard and (ii) whose results do not
change after we add some candidates that the voters rank last. Such rules
include, e.g., Copeland or Maximin~\cite{bredereck2016prices} and,
indeed, both conditions are commonly satisfied (yet, Plurality is an
example of a rule that fails the first criterion, and Veto is an
example of a rule that fails the second one).
\end{remark}

\subsection{Voting Rules and Integer Linear Programs}

In the next section we show that the BSG problem is fixed-parameter
tractable for the parametrization by the number of candidates. Our
algorithm is based on solving an integer linear program (ILP), and we
use the notion of \emph{ILP-expressible voting rules} to capture the
class of rules for which the algorithm is applicable.

Intuitively, we say that a voting rule $\calR$ is ILP-expressible if
the problem of deciding whether a given candidate $p$ is an election winner
can be expressed as a problem of testing whether a certain integer linear
program has a feasible solution. We require that the number of
variables and constraints in this program is a function of the
number of candidates and voter types in the election, and that the
election is specified through variables that represent the number of
voters of different types. Formally, we have the following definition.


\begin{definition}[\bf ILP-expressible voting
  rule]\label{definition:ilp expressible rule}
  Let $E$ be an election with candidate set $C$ and with a collection
  of $n$ voters.
  The preferences of the voters are encoded as a society
  $\mathbf{w} \in \N^{\tau}$, where the $j$-th entry of $\mathbf{w}$
  indicates how many voters of the $j$-th type are present in the
  election.  Let $p \in C$ be a distinguished candidate.
  A voting rule $\calR$ is \emph{ILP-expressible} if there exists a
  computable function $f$ and integers $\tau',r \leq f(\tau+|C|)$, a
  matrix $W \in \Z^{r \times (\tau + \tau')}$, and a vector
  $\mathbf{b} \in \Z^r$ such that (a) $W$ and $b$ are computable in
  FPT time with respect to $m+\tau$, and (b) it holds that $p$ is a
  winner of the election, i.e., $p \in \calR(E)$, if and only if:
  \[
    \exists \mathbf{x} \in \Z^{\tau'}: W \cdot [\mathbf{w}, \mathbf{x}]
    \leq \mathbf{b}, 
  \]
  where $[\mathbf{w}, \mathbf{x}] \in \Z^{\tau + \tau'}$ is the column
  vector obtained by concatenating $\mathbf{w}$ and $\mathbf{x}$.
\end{definition}

The class of ILP-expressible voting rules is quite large.  Indeed, it
includess all scoring rules, all C1 rules (these are rules depending
only on the majority graph of the input election), Bucklin, STV,
Kemeny, and many others. In the two examples below we show arguments
regarding the Borda rule and the STV rule.

\begin{example}[\bf Borda is ILP-expressible]
  We use the same notation as in Definition~\ref{definition:ilp
    expressible rule}.  Additionally, for a candidate $c \in C$ and
  voter type $j \in [\tau]$, by $\textrm{rank}(c,j)$ we mean the
  position on which $c$ is ranked by the type-$j$ voters.
  To show that Borda is ILP-expressible, we describe a collection of
  linear inequalities (defining the matrix $W$), which are satisfied exactly if
  the Borda score of $p$ is at least as large as the score of every
  other candidate:
  \[
    \sum_{\substack{i \in [m], j \in [\tau] \\ \textrm{rank}(c,j) =
        i}}
    (m-i) w_j \leq \sum_{\substack{i \in [m], j \in [\tau] \\
        \textrm{rank}(p,j) = i}} (m-i) w_j \quad \forall c \in C, c
    \neq p \, .
  \]
\end{example}

\begin{remark}\label{remark:or-in-ilp} 
  For our next example, as well as for some further arguments in the
  paper, we will need the ability to express a disjunction of several
  inequalities as an integer linear program. We now discuss how to
  achieve this (we mention that such tricks are well-known regarding
  ILPs~\cite[Section 7.4]{ilptricks}).  For the sake of simplicity,
  let $x_1, x_2, x_3$ be three ILP variables and let us consider four
  linear inequalities (values $a_{i,j}$ and $b_i$ are constants):
  \begin{align*}
    a_{1,1} x_1 + a_{1,2} x_2 + a_{1,3} x_3 &\leq b_1, \\
    a_{2,1} x_1 + a_{2,2} x_2 + a_{2,3} x_3 &\leq b_2, \\
    a_{3,1} x_1 + a_{3,2} x_2 + a_{3,3} x_3 &\leq b_3, \\ 
    a_{4,1} x_1 + a_{4,2} x_2 + a_{4,3} x_3 &\leq b_4. 
  \end{align*}
  Additionally, we assume that $x_1$, $x_2$, and $x_3$ come from some
  bounded domain (this is always the case for the variables that we
  use in our ILP programs, and can be guaranteed in general by standard techniques\footnote{For example, one may solve the (potentially unbounded) continuous relaxation, obtain a continuous optimum $\mathbf{x}^*$, apply the proximity theorem of Cook et al.~\cite{cook1986integer} to determine a box around $\mathbf{x}^*$ which must contain an integer optimum if it exists, and then introduce lower and upper bounds corresponding to this box, hence bounding all variables.}). We would like to express the fact that:
  \begin{align*}
    (\text{the first inequality} \text{ is satisfied}) &\lor 
    (\text{the second inequality is satisfied}) \\
   & \lor (\text{the last two inequalities are satisfied}).
  \end{align*}
  To this end, we introduce three new variables, $y_1$, $y_2$, and
  $y_3$ and the following (in)equalities:
  \begin{align*}
    0 \leq y_i \leq 1  \quad \forall i \in [3], \quad\quad y_1 + y_2 + y_3 = 1.
  \end{align*}
  This ensures that exactly one of these variables takes value $1$ and
  the other ones take value $0$.  Note that
  Definition~\ref{definition:ilp expressible rule} explicitly allows
  using such auxiliary variables. Let $T$ be the largest value that
  either of the left-hand sides of our four ineqalities may take
  (since we know the values $a_{i,j}$ and the ranges of variables
  $x_1$, $x_2$, and $x_3$, we have that the value $T$ is easy to compute). We replace
  our four initial inequalities with the following (note that the last
  two inequalities both involve variable $y_3$):
  \begin{align*}
    (-T+b_1)(1-y_1) &+  a_{1,1} x_i + a_{1,2} x_2 + a_{1,3} x_3 \leq b_1, \\
    (-T+b_2)(1-y_2) &+ a_{2,1} x_i + a_{2,2} x_2 + a_{2,3} x_3 \leq b_2, \\
    (-T+b_3)(1-y_3) &+ a_{3,1} x_i + a_{3,2} x_2 + a_{3,3} x_3 \leq b_3, \\ 
    (-T+b_4)(1-y_3) &+ a_{4,1} x_i + a_{4,2} x_2 + a_{4,3} x_3 \leq b_4. 
  \end{align*}
  By the choice of $T$ and the constraints on the $y$-variables, we
  see that we implemented exactly the desired disjunction. As an
  additional benefit, we can use variables $y_1$, $y_2$, and $y_3$ to
  read off which disjunction clause is satisfied.  
  While, for the sake of readability, our example regards only a few inequalities, it is clear that it can be generalized in a straightforward way.
\end{remark}

\begin{example}[\bf STV is ILP-expressible]
  The Single Transferrable Vote rule (STV) is defined via the following
  iterative process: If some candidate is a majority winner (i.e., is
  ranked first by more than half of the voters), then this candidate
  is declared a winner. Otherwise, the candidate with the lowest
  Plurality score is removed from the election and the next iteration
  starts. If several candidates have the same lowest Plurality score,
  then we use lexicographic tie-breaking (in our case, where
  $C = \{c_1, \ldots, c_m\}$, it means removing the candidate with the
  lowest index).

  \newcommand{\rf}{{\mathrm{first}}}

  To show that STV is ILP-expressible, we will first tackle a simpler
  task. Let $(c_{\pi(1)}, \ldots, c_{\pi(m')})$ be a sequence of $m'$
  ($m' \leq m$) candidates. We say that this is an \emph{elimination
    order} if $c_{\pi(1)}$ is removed in the first iteration,
  $c_{\pi(2)}$ is removed in the second interation, and so on, until
  $c_{\pi(m')}$ who is chosen as a winner in the $m'$-th iteration.
  We will provide a set of linear inequalities that are satisfied if
  and only if $(c_{\pi(1)}, \ldots, c_{\pi(m')})$ is a correct
  elimination order. For each $i \in [m]$ and $t \in [m']$, let
  $\rf(c_i,t)$ be the set of voter types that would rank $c_i$ on the
  first position if candidates $c_{\pi(1)}, \ldots, c_{\pi(t-1)}$ were
  deleted. 
  To
  ensure that the elimination order is correct, we introduce
  the following inequalities:
  \begin{enumerate}
  \item For each $i \in [m'-1]$ we have to ensure that $c_{\pi(i)}$
    has the lowest Plurality score among the candidates remaining in
    the $i$-th round (and that it has the lowest index among the
    candidates with the same Plurality score). For each
    $c_j \in C \setminus \{c_{\pi(1)}, \ldots, c_{\pi(i-1)}\}$ we
    introduce one of the following inequalities. If $j < \pi(i)$ (and,
    thus, $c_{\pi(i)}$ has to have strictly lower Plurality score than
    $c_j$), we have:
    \[
      \sum_{\ell \in \rf(c_j,i)} w_\ell < \sum_{\ell' \in \rf(c_{\pi(i)},i)} w_{\ell'},
    \]
    and if $j > \pi(i)$ then we have an analogous inequality, but with ``$<$'' replaced
    with ``$\leq$.''
  \item For each $i \in [m'-1]$ we have to ensure that neither of the
    remaining candidates has majority support. Thus, for each
    $i \in [m'-1]$ and each
    $c_j \in C \setminus \{c_{\pi(1)}, \ldots, c_{\pi(i-1)}\}$ we have the
    inequality:
    \[
      \sum_{\ell \in ref(c_j,i)}w_\ell \leq \nicefrac{1}{2}\sum_{\ell' \in [\tau]}w_{\ell'}.
    \]
  \item We require that $c_{\pi(m)}$ is selected in the final round,
    so we also have inequality:
    \[
      \sum_{\ell \in ref(c_{\pi{m}},m)}w_\ell > \nicefrac{1}{2}\sum_{\ell' \in [\tau]}w_{\ell'}.
    \]
  \end{enumerate}
  To ensure that $p$ is an STV winner, we use the above approach to
  generate inequalities for every possible elimination order that ends
  with $p$, and---using the trick from
  Remark~\ref{remark:or-in-ilp}---we form their disjunction.

\end{example}

\begin{remark}
  Our class of ILP-expressible rules is very similar to the class of
  election systems \emph{described by linear inequalities} of Dorn and
  Schlotter~\cite{DornS:2012}; their definition says that there must
  exist $f(m)$ linear systems $W_i \mathbf{w} \leq \mathbf{b}_i$,
  $i \in [f(m)]$, such that $p$ is a winner if at least one of these
  systems is satisfied.  In other words, it means that $\mathcal{R}$
  can be described by a bounded disjunction of linear systems.  Since
  such disjunctions can be expressed in ILPs, the
  rules that fit their definition also fit ours.\footnote{Note that we need the variables to be bounded to implement such a disjunction.
  However, as demonstrated previously, boundedness is always achievable by standard techniques and, so, it is not an issue.}  Another related notion, called
  \emph{integer-linear-program implementable rules}, due to
  Faliszewski, Hemaspaandra, and Hemaspaandra~\cite[Definition
  6.1]{multimodeattacks}, is weaker than ours in that it does not
  allow for auxiliary variables $\mathbf{x}$; without these variables
  it is not clear how to define, e.g., Bucklin or STV (for example, it
  is not clear how to implement disjunctions without auxiliary
  variables).
\end{remark}

\subsection{Fixed-Parameter Tractability of BSG}\label{sec:ilp-algo}

We prove that $\calR$-BSG is FPT with respect to the number $m$ of
candidates for any ILP-expressible rule.  The result follows by
formulating $\calR$-BSG as an integer linear program and invoking
Lenstra's famous result~\cite{lenstra1983integer} (which implies that
ILP is FPT with respect to the number of integer variables); it is
arguably quite surprising, since, as it turns out, it is possible to
encode the complete diffusion process using integer variables and
linear constraints.
We first prove the result for the synchronous variant, and then show
how to modify it to work also for asynchronous diffusion.

\begin{theorem}\label{theorem:ilp one}
  Synchronous $\calR$-BSG is fixed-parameter tractable with respect to
  the number $m$ of candidates for every ILP-expressible voting rule
  $\calR$.
\end{theorem}

\begin{proof}
  As a preprocessing phase, we augment the given society graph to have
  exactly $m!$ vertices, one vertex for each possible preference
  order; to this end, we might create some vertices of weight
  zero.\footnote{\label{footnote:one}%
    This phase is needed as we want to consider bribery operations,
    which might in certain cases introduce new preference orders not
    originally present in the election.  To avoid formal difficulties,
    we simply introduce those preference orders beforehand.}  Thus,
  the number of types in the input is $\tau = m!$, and the number of
  voters of type $i$, $i \in [\tau]$, is $w_i$.
%
%
  Let $k$ be the number of steps of the diffusion process;
  Proposition~\ref{proposition:convergence} says that $k \leq \tau$
  and, thus, we simply set $k = \tau = m!$.  For $i \in [\tau]$,
  denote by $N[i]$ the closed neighborhood of $i$ (i.e., the set that
  includes $i$ and all the vertices that are directly connected to it
  by and edge). Similarly, by $N(i) = N[i] \setminus \{i\}$ we denote
  $i$'s open neighborhood.  We use the Iverson bracket notation, i.e.,
  for a logical expression $F$, we write $[F]$ to mean $1$ when $F$ is
  true and to mean $0$ when $F$ is false.

  We construct an ILP with the following variables:
  \begin{enumerate}
  \item For each type $i \in [\tau]$ and each diffusion step
    $\ell \in [k]$, we define an integer variable $x_i^\ell$
    representing the number of voters of type $i$ after $\ell$
    diffusion steps.
  \item For types $i,j \in [\tau]$ we define variables $\beta_{ij}$
    describing the bribery, where $\beta_{ij}$ corresponds to the
    number of voters bribed from being of type $i$ to being of type
    $j$; note that we also consider $\beta_{ii}$, the number of voters
    of type $i$ which are not bribed.
  \item For every $i,j \in [\tau]$ and $\ell \in [k]$, we define a
    binary variable $z_{ij}^\ell$ indicating whether in the $\ell$-th
    step the voters of type $i$ are being assimilated into type~$j$
    (for technical reasons, we also use variables $t^\ell_{ij}$; see
    explanations below).
  \end{enumerate}

  Let $c_{ij}$ be the cost of bribing one voter of type $i$ to become
  a voter of type $j$ (we set $c_{ij}$ to be $\infty$ if $j$ is not
  reachable from $i$).  As our aim is to minimize the cost of bribery,
  the objective of our ILP is to minimize
  $\sum_{i,j} c_{ij} \beta_{ij}$.
  Our ILP constraints are presented in Figure~\ref{ilp}; note that
  some of them are non-linear. Below we discuss their meaning and, for
  the non-linear ones, we explain how they can be encoded within an
  ILP.

\begin{figure}
\begin{align}
&\textstyle\sum_{j=1}^\tau \beta_{ij} = w_i & \forall i \in [\tau] \label{eq:brib1} \\ 
&\textstyle\sum_{i=1}^\tau \beta_{ij} = x^0_j & \forall j \in [\tau] \label{eq:brib2} \\
&z_{ij}^\ell = \left[\textstyle\sum_{a \in N[i]} \frac{1}{2}x_a^{\ell -1}  < \textstyle x_j^{\ell -1}\right] & \forall j \in N(i),\ell \in [k] \label{eq:cond} \\
&\textstyle\sum_{j \in N[i]} z_{ij}^\ell = 1 & \forall i \in [\tau] \label{eq:3} \\
&t_{ij}^{\ell} = z_{ij}^\ell x_i^{\ell -1} & \forall i,j \in [\tau], \ell \in [k] \label{eq:nonlinear} \\
&x_j^\ell = \textstyle\sum_{i \in N[j]} t_{ij}^\ell & \forall j \in [\tau],\ell \in [k] \label{eq:move} \\
&W \cdot [\mathbf{y}, \mathbf{x^k}] \leq \mathbf{b} & \label{eq:voting_rule}
\end{align}
\caption{\label{ilp}Constraints used in the proof of
  Theorem~\ref{theorem:ilp one}. We omitted the simple constraints
  requiring that the variables are in the right domains for clarity.}
  \end{figure}

  \begin{description}
  \item[ Constraints~{(\ref{eq:brib1})} and~(\ref{eq:brib2}).]  These
    constraints are standard and express that the vector
    $\mathbf{x}^0$ describes the society after the bribery (recall
    that $\beta_{ii}$ corresponds to non-bribed voters of type
    $i$).

  \item[Constraint~(\ref{eq:cond}).] This constraint assigns $1$ to
    $z_{ij}^\ell$ if the weight in type $j$ exceeds half of the total
    weight of $N[i]$ and $0$ otherwise.  Note that we do \emph{not}
    affect $z_{ii}^\ell$ here as the constraint goes only over $j$ in
    the open neighborhood $N(i)$.  Since Constraint~(\ref{eq:cond}) is
    non-linear, using the approach from Remark~\ref{remark:or-in-ilp},
    for each $i$, $j$, and $\ell$ we express it as a disjunction of
    two inequalities:
    \[
      \left[\textstyle\sum_{a \in N[i]} \frac{1}{2}x_a^{\ell -1}  < \textstyle x_j^{\ell -1}\right] \lor
      \left[\textstyle\sum_{a \in N[i]} \frac{1}{2}x_a^{\ell -1}  \geq \textstyle x_j^{\ell -1}\right].
    \]
    The approach taken in Remark~\ref{remark:or-in-ilp} provides us
    with a binary variable, which takes value $1$ if the first
    inequality is satisfied and value $0$ otherwise. We simply take
    $z^\ell_{ij}$ to be this variable.

  \item[Constraint~(\ref{eq:3}).] This constraint enforces that at
    least one of $z_{ij}^\ell$ is $1$, and this includes
    $z_{ii}^\ell$; thus, if there is no $j \in N(i)$ with weight more
    than half of the weight of $N[i]$, then $z_{ii}^\ell=1$ holds,
    which corresponds to $i$ keeping its weight (i.e., voters of type
    $i$ are not being assimilated into some other type).

  \item[Constraints~(\ref{eq:nonlinear}) and~(\ref{eq:move}).] These
    constraints define the weights for step $\ell$, given the weights
    from step $\ell - 1$.  Precisely, $x_j^\ell$ takes the weight of
    all its neighbors (including itself) for whom $z_{ij}^\ell=1$.  We
    use the $t_{ij}^{\ell}$ variables as temporary variables that are
    non-zero for those $i$ and $j$ for which $z_{ij}^\ell = 1$.
    Notice that Constraint~\eqref{eq:nonlinear} is non-linear but,
    again, can be handled using standard ILP tricks. Indeed, we can
    express the constraint as:
  \[
    \left((z_{ij}^\ell = 1) \implies (t_{ij}^l =
      x_{i}^{\ell-1})\right) \wedge \left((z_{ij}^\ell = 0) \implies
      (t_{ij}^l = 0)\right)\ .
  \]
  Since variables $z_{ij}^\ell$ are binary and implication can be seen
  as a disjunction, we can further transform this into:
  \[
    \left((z_{ij}^\ell = 0) \lor (t_{ij}^l = x_{i}^{\ell-1})\right)
    \wedge \left((z_{ij}^\ell = 1) \lor (t_{ij}^l = 0)\right)\ .
  \]
  We express the disjunctions as in Remark~\ref{remark:or-in-ilp}.
  Everything else is linear.

\item[Constraint~(\ref{eq:voting_rule}).] This constraint corresponds
  to the specific voting rule being considered, with $\mathbf{y}$ as
  the auxiliary variables (called $\mathbf{x}$ in
  Definition~\ref{definition:ilp expressible rule}); it is satisfied
  if and only if the given, preferred candidate $p$ wins the election
  specified by $\mathbf{x^k}$ (i.e., the election after the bribery
  and at the end of the diffusion process). The constraint can be
  expressed as part of an ILP due to our assumption that $\calR$ is
  ILP-expressible.
  \end{description}
  This completes the description of our ILP. As the number of
  variables is a function of the number of candidates, we solve it in
  FPT time using the algorithm of Lenstra~\cite{lenstra1983integer}.
\end{proof}


\begin{corollary}
  Asynchronous $\calR$-BSG is fixed-parameter tractable with respect
  to the number $m$ of candidates for every ILP-expressible voting
  rule $\calR$, for both the optimistic and the pessimistic variants.
\end{corollary}

\begin{proof}
  For the optimistic variant, we modify the ILP described above as
  follows.  We add variables $y_i^\ell$, representing whether type $i$
  is updated in the $\ell$-th step. We require that
  $\sum_{i} y_i^\ell = 1$ to enforce that exactly one vertex is
  updated.  We add variables $\hat{z}_{ij}^\ell$ and we want to
  enforce that $\hat{z}_{ij}^\ell = z_{ij}^\ell \wedge y_i^{\ell}$;
  the interpretation is that $i$ might be assimilated only if
  $y_i^\ell = 1$. However, instead of directly encoding
  $\hat{z}_{ij}^\ell = z_{ij}^\ell \wedge y_i^{\ell}$, in our case it
  suffices to add constraints
  $\hat{z}_{ij}^\ell \leq z_{ij}^\ell, \, \hat{z}_{ij}^\ell \leq
  y_i^\ell$.  Then, it suffices to replace $z_{ij}^\ell$ with
  $\hat{z}_{ij}^\ell$ in constraint~\eqref{eq:move}.
  Finally, we need to enforce that the diffusion ends after at
  most $k$ steps. To this end, for each pair of types $i$, $j$ we add
  a constraint requiring that it is impossible to assimilate $i$ into
  $j$ after the $k$-th diffusion step (the constraint can be expressed analogously to
  constraint~\eqref{eq:cond} and we do not spell it out explicitly).

  For the pessimistic variant, notice that any sequence of diffusion
  steps which converges contains a maximal irredundant sequence, and
  irredundant sequences are of length at most $\tau$ (this follows from
  Proposition~\ref{proposition:convergence}).  It thus suffices to
  consider the set of permutations of $[\tau]$. Expressing that in
  none of them $p$ is losing can be done by a long conjunction of ILPs
  given by constraints~\eqref{eq:cond}-\eqref{eq:voting_rule}, with a
  clause for each sequence (note that we could have used analogous
  approach for the optimistic variant too, but the presented approach is more
  efficient).
\end{proof}

%

\section{Model Generalizations}\label{section:model extensions}

Here we generalize the simple model described above and demonstrate
far broader scenarios for which $\calR$-BSG remains fixed-parameter
tractable. In particular, we consider models with arbitrary connections between voter types, models with different bribery operations and manipulative actions, and models with different diffusion processes.

\subsection{Various Voter Types}

Instead of partitioning the voters by preference orders, we can
consider arbitrary partitions.  Our motivation may be either that the
partition by preference orders is too crude when there are significant
differences between voters with the same preference (e.g., young and
old voters are convinced by different methods), or, on the other hand,
the partition may be too fine-grained when different preference orders should
nonetheless be treated identically, e.g., because of the choice of a
voting rule which does not significantly distinguish them.  As the
number of variables in the ILP described in the proof of
Theorem~\ref{theorem:ilp one} depends only on the numbers of types and
candidates, it follows that $\calR$-BSG remains fixed-parameter
tractable with respect to the number~$\tau$ of types and the number of
candidates.
Taken to the extreme, namely if we set each voter in a given election
to constitute her own voter type, we arrive at the model of diffusion
studied, e.g., by Wilder and Vorobeychik~\cite{wilder2017controlling}
or by Bredereck and Elkind~\shortcite{bredereck2017manipulating}.

\subsection{Arbitrary Bribery Operations and Manipulative Actions}

Our model can incorporate, e.g., all bribery operations mentioned by
Faliszewski and Rothe~\shortcite{briberysurvey}. Indeed, the constants
$c_{ij}$ used in the proof of Theorem~\ref{theorem:ilp one} encode the
cost of transforming a voter of type $i$ into a voter of type $j$ and
can be redefined for other bribery operations.

Furthermore, following the discussion of Knop et
al.~\shortcite[Section 3.2]{bestpaperever}, this approach can be
extended to other types of manipulative operations, such as
\emph{voter control}~\cite{briberysurvey}, at no asymptotical cost in
terms of computational complexity.  Specifically, say that some voters
are active while others are latent, and there is a cost for activating
a latent voter or vice versa.  Accordingly, we can define two types
for each preference order, one corresponding to it being active while
the other corresponding to it being latent, and set the cost of
``moving'' a voter from the active type to the latent and/or vice
versa according to the specific control operation.  Notice that this
only doubles the number of voter types, so the asymptotic complexity
remains intact.

\subsection{General Diffusion Processes}
\label{subsec:extensions}

So far our society graphs had an undirected edge between two vertices
if their corresponding preference orders were of swap distance one,
and we considered a specific, simple diffusion process.  In fact, our
model can incorporate directed arcs, where a vertex would be
influenced by those vertices for which it has an outgoing arc and, in
particular, we do not have to be confined to connections between types
associated with preferences that differ in the ranking of a single
pair of candidates.
Furthermore, those arcs can be weighted,
representing different influence strengths 
(e.g., consider damping the influence of voters which are, swap distance-wise, farther).
Adding weights can be done by modifying Equation~\eqref{eq:cond} in a straightforward way.

Moreover, and most importantly, we can express in our model a large
class of diffusion processes.  The following definition is inspired by
viewing the diffusion of preferences as an abstract process, in which
each voter holds a local election to decide which preference order to
assume.  For example, the diffusion process described in
Section~\ref{section:preliminaries} corresponds to holding an election
containing the voters of swap distance at most one, and changing to
the preference order of the majority, if such exists.
%
In the definition below we use the notation from Theorem~\ref{theorem:ilp one}. 

\begin{definition}[\bf ILP-expressible diffusion
  process]\label{def:ilp expressible diff}
  Let $k$ be an upper bound on the number of diffusion steps, recall
  that for diffusion step $\ell \in [k]$, the variables
  $x_i^{\ell - 1}$, $i \in [\tau]$, express the current society, and
  let $f$ be a computable function.  Then, an \emph{ILP-expressible
    diffusion process} is a process such that for each
  $i,j \in [\tau]$ and $\ell \in [k]$, there are integers
  $r(i,j, \ell), \tau(i,j,\ell) \leq f(\tau)$, a matrix
  $D_{i, j, \ell} \in \Z^{r(i,j,\ell) \times \tau + \tau(i,j,\ell)}$,
  and a vector $\mathbf{b}_{i,j,\ell} \in \Z^{r(i,j,\ell)}$ such that,
  in the $\ell$-th diffusion step, voters of type $i$ are assimilated
  into type $j$ if and only if the following formula is satisfied:
$$\exists \mathbf{x}' \in \Z^{\tau(i,j,\ell) } \quad D_{i,j,\ell}(\mathbf{x}', \mathbf{x}^{\ell - 1}) \leq \mathbf{b}_{i,j,\ell}\ .$$
\end{definition}

Our basic diffusion process corresponds to Equation~\eqref{eq:cond}.
Another ILP-expressible diffusion process is that each voter replaces
her preference order by the Kemeny ranking computed for the voters in
her neighborhood (for the definition of a Kemeny ranking, see the
original research paper~\cite{kem:j:no-numbers} or, e.g., the chapter
of Fischer et al.~\cite{fis-hud-nie:b:tournaments}).




\begin{remark} \label{rem:noconvergence}
  Proposition~\ref{proposition:convergence} does not hold for all
  generalized diffusion processes, as the number of diffusion steps
  might not be bounded by the size of the society graph or the diffusion
  may never stabilize (cf. Remark~\ref{ex:periodic} below).  (Also,
  new voter types might sometimes appear as a result of diffusion
  steps.)  Thus, the corresponding ILP to solve $\calR$-BSG would have
  to be supplied with the number $k$ of diffusion steps to simulate.
  Sometimes it is indeed plausible that an agent
  can estimate the number of diffusion steps to occur after the
  manipulative actions, e.g., when he knows the time of the election,
  or when it is provable (although differently than by the argument
  of Proposition~\ref{proposition:convergence}) that the process
  stabilizes after $k$ steps.

\end{remark}

\begin{theorem} \label{thm:bsg-generalization} $\calR$-BSG is
  fixed-parameter tractable with respect to the number of candidates,
 the number $\tau$ of types, and the number $k$ of diffusion steps
  if both $\calR$ and the diffusion process are ILP-expressible.
\end{theorem}

\subsection{Exemplary Models}

We conclude this section with several examples of scenarios that are
captured by such generalized diffusion processes.

\begin{example}[\bf Multidimensional Societies]
  Consider voters of different age groups.  It is plausible that the
  tendency to be influenced by other voters depends on age, and so we
  might
  have a voter type for each tuple of (preference order, age group),
  with different outgoing arcs and different diffusion conditions.
  Note that we can thus express, e.g., that voters may be strongly
  influenced by some voter groups, yet cannot be assimilated into them
  (e.g., a junior person can be influenced by a senior one, but this
  does not make him or her senior).
\end{example}

\begin{remark}[\bf Periodic behavior] \label{ex:periodic} As an
  example of how Proposition~\ref{proposition:convergence} might not
  hold (as noted in Remark~\ref{rem:noconvergence}), consider that
  there are three voter types ($0,1,2$) subdivided into young (Y) and
  old (O).  Moreover, let's say that young people change their mind,
  but old people don't.  Next, people of type $i$ are influenced by
  people of type $i+1 \mod 3$.  Denote $w_{i,T}$ for $i \in \{0,1,2\}$
  and $T \in \{Y, O\}$ the number of people of type $i$ and age group
  $T$.  Finally, say that $\min_i w_{i,O} \geq \max_i w_{i,Y}$, i.e.,
  there are at least as many old people of each type as there are
  young people of any type.  Consider the synchronous model. Then, in
  each iteration, the young people of type $i$ move to type
  $i+1 \mod 3$, and the diffusion does not stabilize but oscillates
  periodically.  (Similar behavior can be achieved also in the
  asynchronous model.)
\end{remark}

\begin{example}[\bf Advanced diffusion model]
  Say there are four candidates $\{c_1, \dots, c_4\}$,
  $3$~age groups $\{Y, M, O\}$ (for Young, Middle-aged, and Old), and
  two ``stubbornness'' levels $\{P, S\}$ (for Persuadable and
  Stubborn).  People are divided into types by their preference, age
  group, and stubbornness level. Since there are $4! = 24$ possible
  preference orders, there are altogether
  $\tau = 3 \cdot 2 \cdot 24 = 144$ possible voter types.  Next, we
  describe ``coefficients of influence'' which will be used to define
  the diffusion process.  Define a function $f$ as follows:
  \begin{align*}
    &f(Y,Y) = 1.2, &
    &f(M,Y) = 0.8, &
    &f(M,M) = 1, \\
    &f(O,M) = 0.5, &
    &f(M,O) = 0.3, & 
    &f(O,O) = 0.8, \\
    &f(O,Y) = 0,&
    &f(Y,O) = 0,&
    &f(Y,M) = 0.
  \end{align*}
  The meaning is that age group $G_1$ influences age group $G_2$ with
  a coefficient $f(G_1, G_2)$ for the persuadables, and with a
  coefficient $0.5 \cdot f(G_1, G_2)$ for the stubborn.  For example,
  persuadable young people have relatively low self-esteem and weight
  the opinion of other young people as $1.2$ higher than their own;
  they weigh the opinion of middle-aged people as $0.8$ of their own,
  and they completely disregard the opinions of old people.
	
  Denote by $G(t)$ the age group of type $t$, let $S(t)=1$ if $t$ is
  persuadable and $S(t) = 0.5$ otherwise, and let $k \in \N$ be the
  number of rounds of diffusion to simulate, perhaps as an estimate of
  how much diffusion happens before a (global) election takes place.
  We assume that one type is influenced by other types to a degree
  that exponentially decreases with the swap distance of their
  preference orders, and also that it decreases inversely as the
  election draws closer, perhaps because voters become skeptical and
  more rigid in their opinions.  The coefficient of influence of voter
  type $t$ on voter type $t'$ in round $\ell \in [k]$ is computed as
  follows:
  \[
	c_{t, t'}^\ell = 
	\begin{cases}
	\nicefrac{1}{\ell} \cdot \nicefrac{1}{2^{d}} \cdot S(t') \cdot f(G(t), G(t')) & \text{if } t' \neq t\ , \\
	1 \enspace & \text{if } t=t'\ .
	\end{cases}
  \]
  where $d$ is the swap distance between the preference orders of
  voters of types $t$ and $t'$.

  Then, the diffusion process is synchronous and in each step, each
  voter type $t'$ holds a ``local election'' defined as follows. There
  are $c(t,t') \cdot w_t$ many voters of type $t$ (for all types
  $t \in [\tau]$), and the voting rule is Borda (note that for Borda,
  the definition is sensible even if the number of voters as defined
  is fractional).  After this election, we obtain some winning ranking
  $r$ (breaking ties according to the ordering of candidates as
  $c_1, \dots, c_4$), and all voters of type $t'$ move to a type with
  ranking $r$ and the same age group and stubbornness level.  The
  final vote is evaluated with the Plurality rule.
	
  Now we wish to apply Theorem~\ref{thm:bsg-generalization} and for
  that we need to show that $\calR$ and the diffusion process are
  ILP-expressible.  Since $\calR=$ Plurality, we focus on the
  diffusion process.  Let $\text{rank}(i,t)$ be the rank of candidate
  $i$ in the preference order of type $t$.  We define the score of
  candidate $i$ in the local election of type $t$ in round $\ell$ as
  \[
	s_{i,t}^\ell = \sum_{t' \in [\tau]} c_{t',t}^\ell (4- \text{rank}(i,t')) x_{t'}^{\ell-1} \enspace .
  \]
  Recall that we use the Iverson bracket notation, i.e., $[F]$
  evaluates either to $1$ or to $0$, depending on the truth value of
  condition $F$.  For each two $i, i' \in [4]$, $i \neq i'$, define:
  \[
	i \triangleleft^\ell_t i' \equiv \left[\left(s_{i,t}^\ell > s_{i',t}^\ell\right) \vee \left( (s_{i,t}^\ell = s_{i',t}^\ell) \wedge i > i'\right) \right] \enspace .
  \]
	
  Definition~\ref{def:ilp expressible diff} requires that, for each
  two types $t,t'$ and each round $\ell \in [k]$, there is a linear
  system that is satisfied exactly if voters of type $t$ should be
  assimilated by type $t'$.  Let $t'$ have a preference order
  $c_{i_1} \preceq c_{i_2} \preceq c_{i_3} \preceq c_{i_4}$ and let
  $G(t) = G(t')$ and $S(t) = S(t')$.  Then $t$ should be assimilated
  into $t'$ exactly if
  \[(i_1 \triangleleft^\ell_t i_2) \wedge (i_2 \triangleleft^\ell_t i_3) \wedge (i_3 \triangleleft^\ell_t i_4)\ ,
  \]
  which is a boolean combination of linear inequalities and can be
  rewritten into the format required by Definition~\ref{def:ilp
    expressible diff} using the same standard tricks used in the proof
  of Theorem~\ref{theorem:ilp one} (see~\cite[Section
  7.4]{ilptricks}).  This shows that the diffusion process is
  ILP-expressible.  Hence, if the above were generalized to $m$
  candidates, $\alpha$ age groups, and $\sigma$ stubornness levels,
  then $\calR$-BSG would be fixed-parameter tractable with respect to
  $m+\alpha+\sigma+k$ for any ILP-expressible voting rule $\calR$.
\end{example}

\section{Experiments}\label{section:experiments}

In addition to our theoretical results, we also evaluated our
ILP-based algorithm experimentally. Unfortunately, it turned out that, while it can produce results for up to four candidates in a reasonable
amount of time, going beyond this number is not practical. Thus, we
sought heuristic algorithms instead. In particular, we designed one deterministic heuristic, a greedy algorithm, and one heuristic based on simulated annealing. Unfortunately,
our other attempts were not very successful either---our heuristics often produced
much more costly bribery strategies than the (optimal) ILP algorithm (even for
the cases of $3$ or $4$ candidates) and often required even more time
to complete (especially for the cases with more voters).  Two possible
explanations for these results are that:
\begin{enumerate}
\item Our $\calR$-BSG problem is a particularly hard combinatorial problem.
  If this is indeed the case, then it might be a good testbed for improved
  ILP algorithms.
\item Our heuristics are poorly designed and there is room for
  significant improvement and further research to obtain better ones. Indeed, we did not aim at optimizing heuristics but rather at gathering a basic feeling of the practical complexity of our problem.
\end{enumerate}

In either case, our results call for further research and further
analysis. Below we describe our experimental setup, the heuristics
that we have tried, and how they compare to the ILP-based algorithm.

\subsection{Experimental Setup}

We consider elections with either $3$ or $4$ candidates and either
$1000$ or $10000$ voters. In each case, we generate the voters'
preferences using the impartial culture model; i.e., by drawing the
preference order of each voter uniformly at random.  We used
the Borda voting rule and focused on the basic variant of our problem, where voter types are equivalent to the voters' preference orders, two
voter types are connected if their swap distance is one, and there is
unit cost for shifting the preferred candidate by one position up in a
single voter's preference order. We considered the synchronous
diffusion process only.
For the case of three candidates, we generated $55$
elections for each combination of a heuristic algorithm and the number of voters.
For the case of four candidates, we generated $10$ elections for
each heuristic algorithm. We ran the ILP-based algorithm for every
generated election, in each setting.




\subsection{Algorithms}

We tested three algorithms, namely our ILP-based algorithm---described in Section~\ref{sec:ilp-algo}---and two heuristics. The ILP-based algorithm
solves the optimization variant of $\calR$-BSG, i.e., it does not need
the budget to be part of its input; it simply finds the lowest cost of
ensuring that the preferred candidate is a winner. The heuristics, on
the other hand, are phrased as decision algorithms and, thus, need the
budget $b$ to be given. We convert them to optimization algorithms
using the standard approach of binary searching for the right value of
$b$.
Specifically, our binary-search method has two phases and, for a given
decision algorithm $\mathcal{H}$, proceeds as follows:
\begin{enumerate}
\item In the \emph{first phase}, we begin with $b = 1$; then, we use
  algorithm $\mathcal{H}$ to see whether this budget is sufficient to
  succeed. If this is the case, then we halt; otherwise, we
  multiply~$b$ by~$2$ and, again, use $\mathcal{H}$ to check whether
  this $b$ suffices.  If this is the case, then we halt; otherwise, we
  again multiply $b$ by~$2$ and repeat.  The first phase continues
  until we have some $b = 2^i$ for which $\mathcal{H}$ succeeds.
  
\item In the \emph{second phase}, we perform a binary search on the
  values of $b$ between $b=2^{i-1}$ and $b=2^i$; this allows us to
  find the minimum $b^*$ for which $\mathcal{H}$ succeeds; we return
  this $b^*$.
\end{enumerate}

\subsection{Heuristic Algorithms}

We design two heuristic algorithms. One is a simple greedy algorithm
whereas the other is an adaptation of the classic simulated annealing
approach. Both algorithms are using the same idea regarding evaluation
of partial solutions, based on the idea of \emph{margin of
  victory}. Given a (partial) solution for the problem, i.e., the
number of positions by which to shift the preferred candidate in the
preference orders of the voters, we evaluate the quality of this
solution as follows:
\begin{enumerate}
\item we implement the shifts as specified in the solution,
\item we run the diffusion process, and
\item we compute the difference between the Borda score of the
  preferred candidate and the Borda score of the highest-scoring
  opponent; this value is known as margin of victory and we interpret
  it as the quality of the solution.
\end{enumerate}
Note that if the solution indeed leads to the victory of the preferred
candidate then the margin of victory is non-negative. Either way, the
higher it is, the better (indeed, if it is positive then we want the
preferred candidate to have high advantage over the second-best
candidate; if it is negative, then we want the preferred candidate to,
nonetheless, be as close as possible to the current winner).

\subsubsection{The Greedy Heuristic}

Our greedy heuristic proceeds as follows: We maintain a solution, which
is a set of bribery operations, initialized to be the empty set. We
perform $b$ iterations, where in each iteration we go over all
possible bribery operations---one operation at a time---and select the
operation that, when added to the partial solution, increases its
quality the most.
We also experimented with several other variants of this heuristic,
but neither of them led to substantial changes or improvements in the
performance.






\begin{figure}
	\centering
		\includegraphics[width=12.0cm]{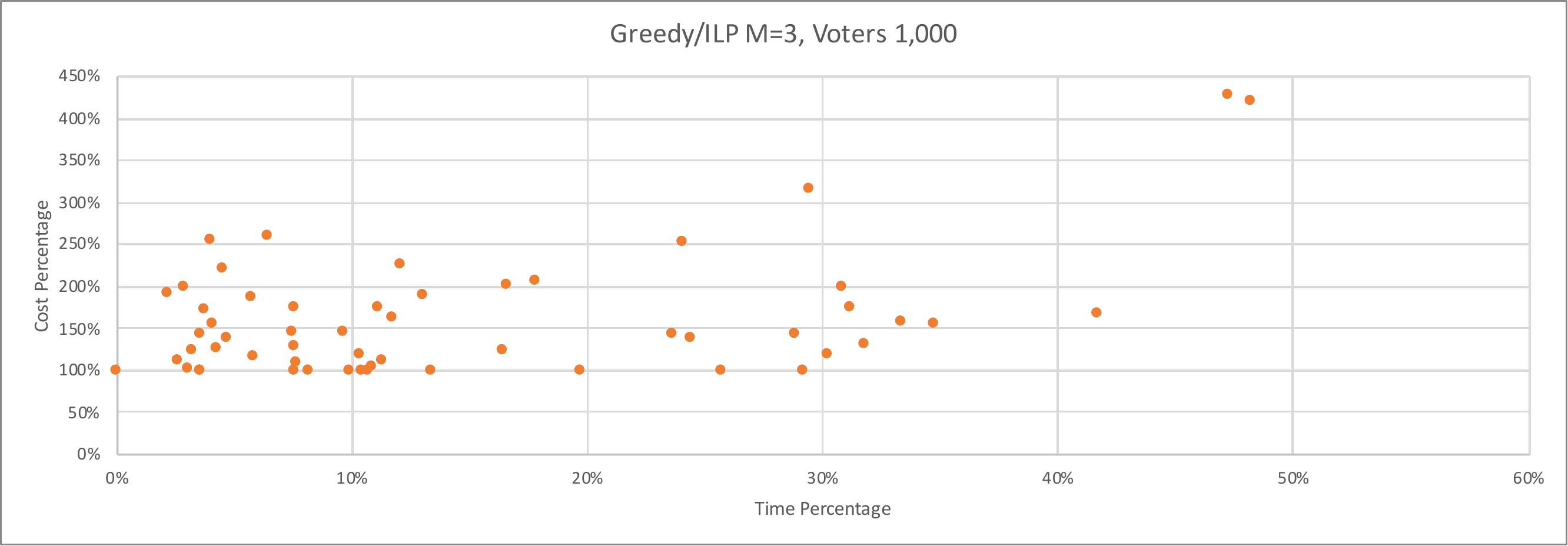} 
	\caption{The performance of the greedy heuristic against ILP with respect to run time and cost with 3 candidates and 1000 voters.}
	\label{m3_1000_zoom_a.pdf}
        \bigskip
        
	\centering
		\includegraphics[width=12.0cm]{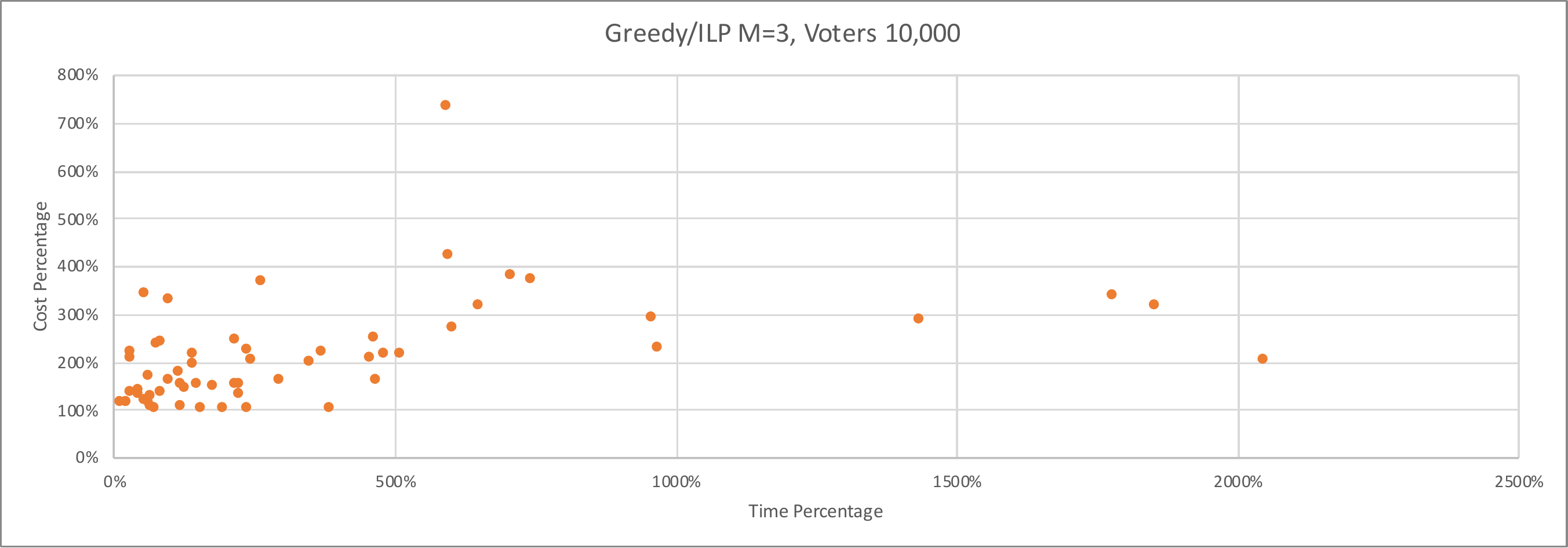} 
	\caption{The performance of the greedy heuristic against ILP with respect to run time and cost with 3 candidates and 10000 voters.}
	\label{m3_10000_zoom_a.pdf}
\end{figure}


\begin{figure}
	\centering
		\includegraphics[width=12.0cm]{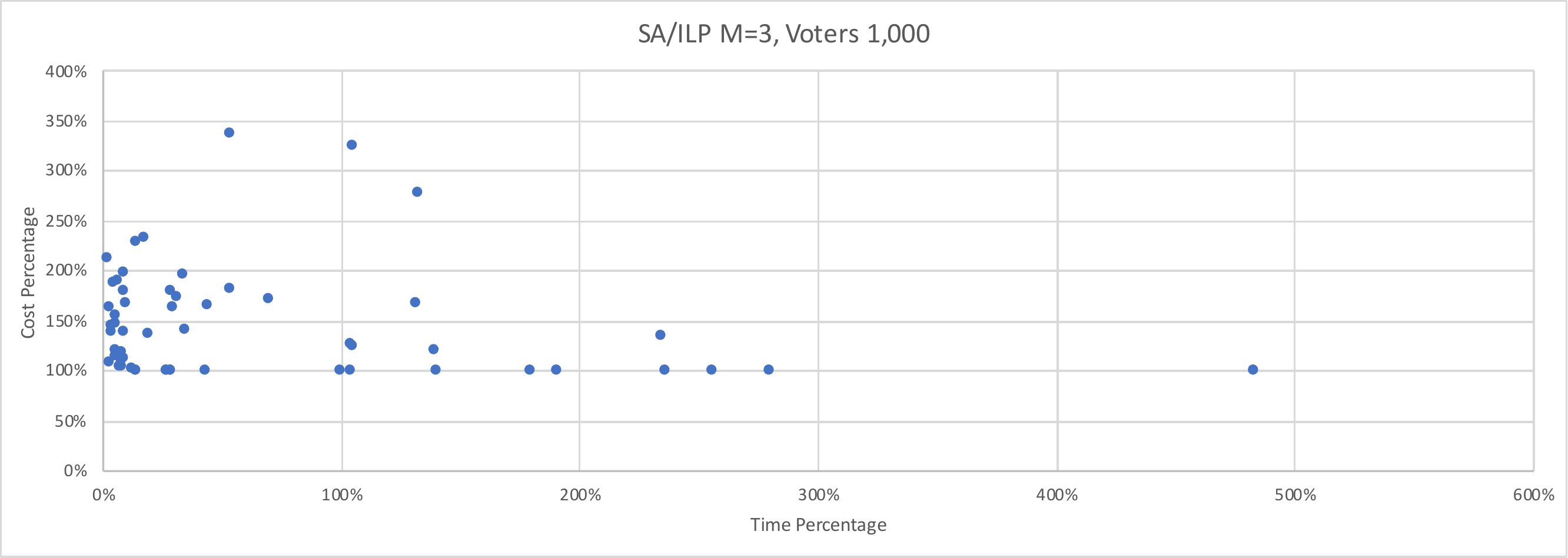} 
	\caption{The performance of the SA heuristic against ILP with respect to run time and cost with 3 candidates and 1000 voters.}
	\label{m3_1000_zoom_b.pdf}
   \bigskip
	\centering
		\includegraphics[width=12.0cm]{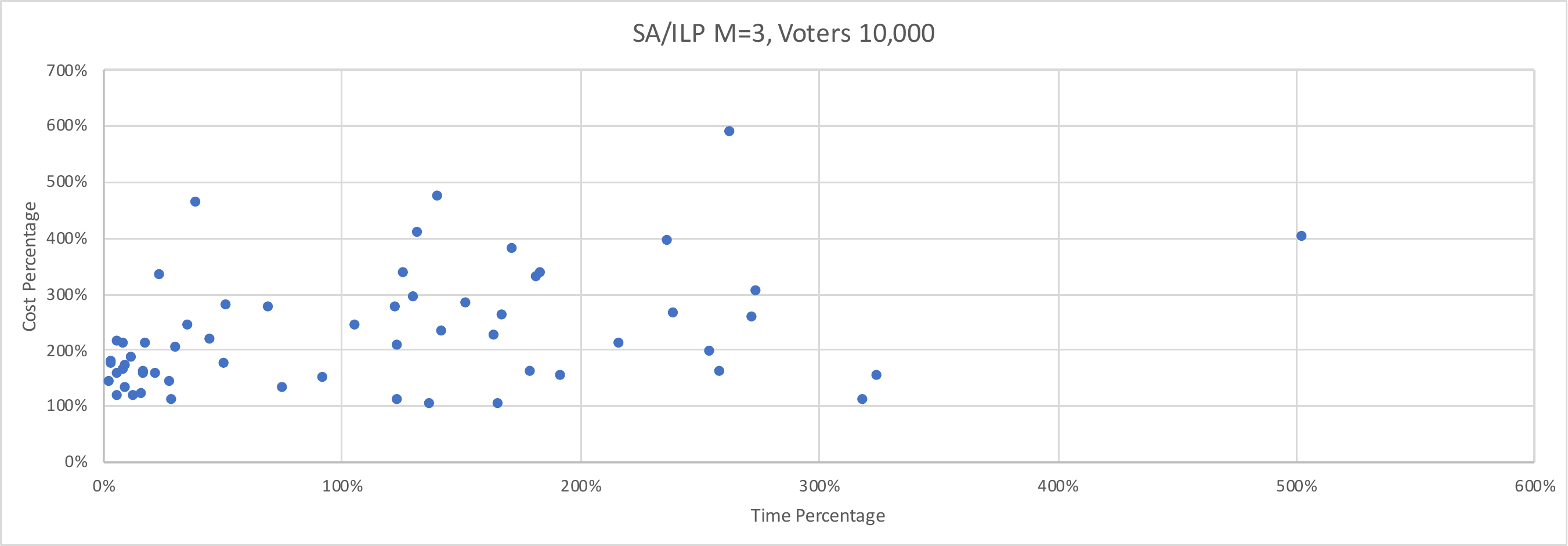} 
	\caption{The performance of the SA heuristic against ILP with respect to run time and cost with 3 candidates and 10000 voters.}
	\label{m3_10000_zoom_b.pdf}
\end{figure}




\begin{figure}
	\centering
		\includegraphics[width=12.0cm]{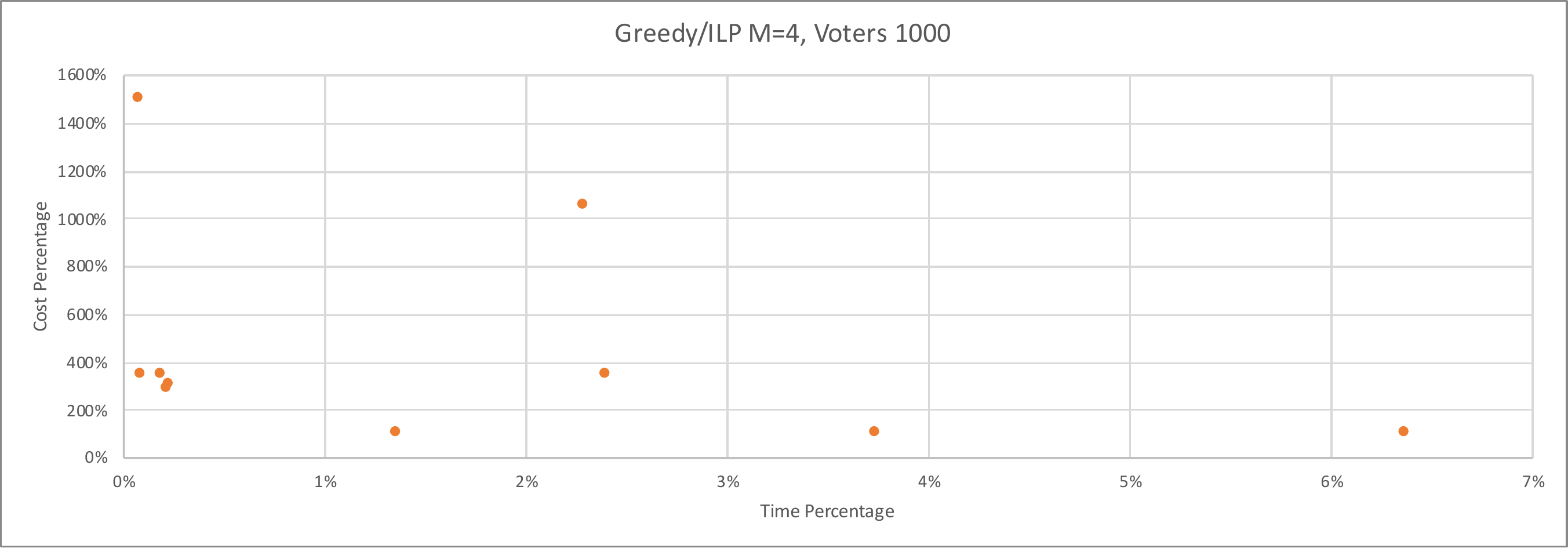} 
	\caption{The performance of the greedy heuristic against ILP with respect to run time and cost with 4 candidates and 1000 voters.}
	\label{m4_1000_zoom_a.pdf}
  \bigskip
	\centering
		\includegraphics[width=12.0cm]{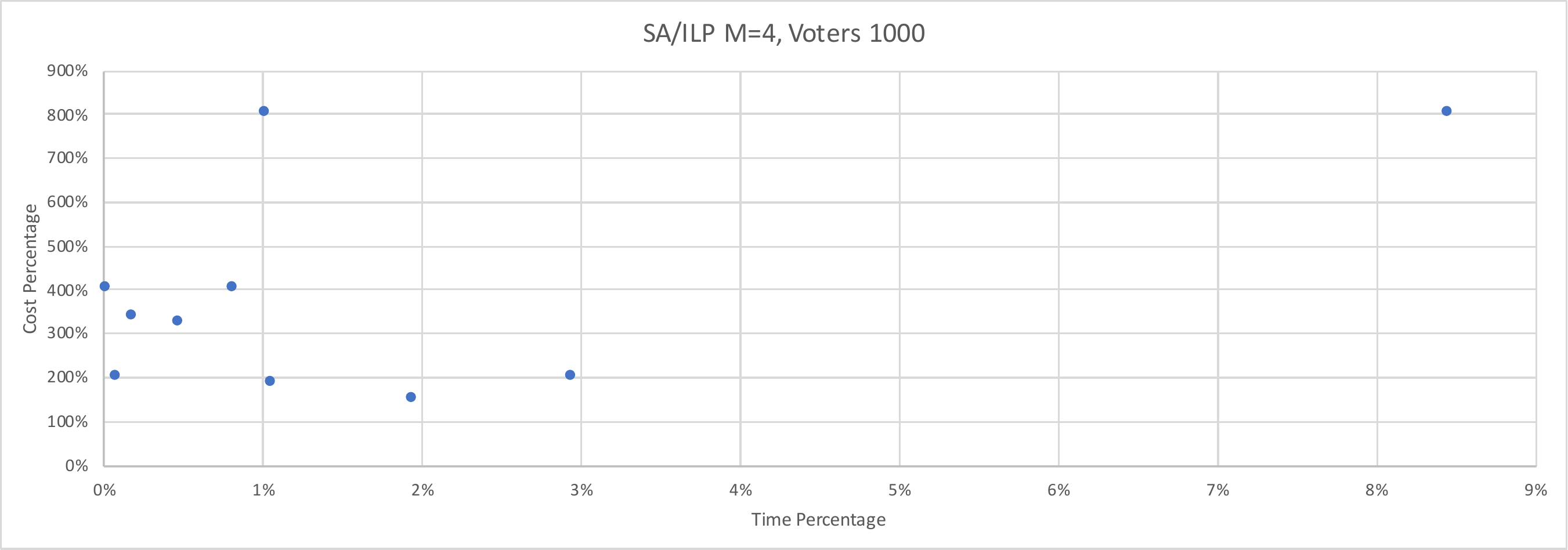} 
	\caption{The performance of the SA heuristic against ILP with respect to run time and cost with 4 candidates and 1000 voters.}
	\label{m4_1000_zoom_b.pdf}
\end{figure}


\subsubsection{The Simulated Annealing Heuristic}

In the simulated annealing heuristic (SA) we maintain a matrix $A$ with $m!$ rows and $m$ columns, which represents $b$ bribery operations, where each cell 
$a_{i,j}$ of the matrix represents the number
of voters from type $i$ for whom we shift the preferred candidate $p$
up by $j$ positions. Note that the cost of the set of bribery operations corresponding to some matrix $A$ is the summation over the costs given by the cells, where the corresponding
cost of cell $a_{i,j}$ is $a_{i,j} \cdot j$ (it costs $j$ to shift $p$
by $j$ positions); for convenience, for a given $A$, we refer to $\sum_{i,j} a_{i,j} \cdot j$ as ``the cost of $A$''.

\paragraph{Initial solution}
We initialize $A$ as follows:
  first, we set all cells of $A$ to be $0$. Then, we iterate until the cost of $A$ is $b$, where in each iteration:
  \begin{enumerate}
    \item 
    We choose $i$ and $j$ uniformly at random ($i$ ranges from $0$ to $m!$, while $j$ ranges from $0$ to $m$). 
    \item 
    If $a_{i,j}$ is $0$, then we select a value $v$ uniformly at random between $0$ and the number of voters of type $i$. 
    \item 
    If the cost of $A$ plus $v \cdot j$ is not greater than $b$, then we set $a_{i,j}$ to be $v$; otherwise, we leave $a_{i, j}$ to be $0$ and proceed to the next iteration.
  \end{enumerate}

\paragraph{Local improvements}
After initializing the matrix $A$ as described above, we proceed to the main algorithm, in which
we perform $T = 10000$ iterations, where in each iteration we aim at improving the current solution.
In particular, in each iteration we pick,
uniformly at random, two different voter types $i \neq k$ (note that we only
choose such voter types for which at least one voter exists) and two indices $j_1$ and $j_2$ such that $a_{k, j_1} > 0$, $a_{i, j_2} > 0$, and proceed as follows:
  (1) we decrease $a_{k,j_1}$ by one;
  (2) we increase $a_{k,j_1-1}$ by one;
  (3) we decrease $a_{i,j_2-1}$ by
one;
  and
  (4) we increase $a_{i,j_2}$ by one.
(Intuitively, (1) and (2) ``free'' one unit of budget as it corresponds to taking one voter of type $k$ that previously was shifted by $j_1$ positions to now be shifted by only $j_1 - 1$ positions; and, then, (3) and (4) ``use'' this one unit of budget as it corresponds to taking one voter of type $i$ that previously was shifted by $j_2 - 1$ positions to now be shifted by $j_2$ positions. Indeed, we choose such $j_1$ and $j_2$ for which this local step is feasible.

\paragraph{Acceptance probability}
Let $A$ be the current solution and let $A'$ be a modified solution according to the procedure described in the previous paragraph.
Next we describe under which conditions we ``accept'' the modification and replace $A$ with $A'$.
First, recall that the number of iterations is $T = 1000$; moreover, we have a parameter $p_0$, set to be $0.2$, and we maintain a parameter $p_1$, initially set to be $p_0$.

Now, if $A'$ has higher quality than $A$, then with probability $1-p_1$ we accept it (i.e., set $A = A'$ and reiterate); while if $A'$ has lower quality than $A$, then with probability $p_1$ we accept it.
Lastly, we replace the value of $p_1$ with $p_1 - p_0 / T$. This ensures that after $T$ iterations $p_1$ drops down to $0$, at which point we halt.

\subsection{Results of the Experiments}

We present the results of our experiments regarding three candidates
in Figures~\ref{m3_1000_zoom_a.pdf} and \ref{m3_10000_zoom_a.pdf} (for
the greedy heuristic) and in Figures~\ref{m3_1000_zoom_b.pdf} and
\ref{m3_10000_zoom_b.pdf} (for the SA heuristic).
For the case of four candidates, we show our results in
Figures~\ref{m4_1000_zoom_a.pdf} and~\ref{m4_1000_zoom_b.pdf}.  For
the case of three candidates, we show results for 1000 and 10000
voters, whereas for the case of four candidates we give results for
1000 voters only.

For the case of three candidates and 1000 voters, the greedy heuristic
is significantly faster than the ILP-based algorithm, but for 10000
voters it works much more slowly. The SA heuristic tends to be slower
than the ILP-based algorithm irrespective of the number of voters.
However, for the case of four candidates both heuristics are much
faster than the ILP solution. Unfortunately, in each of the presented
cases the heuristics quite often return solutions of cost that is much
higher than the optimal one (indeed, for the case of the greedy
algorithm and four candidates, for one of the elections the heuristic
found a solution about 15 times more expensive than the optimal one).
More commonly, the heuristics provide solutions that are up to 3-4
times more expensive than the optimal ones. All in all, this means
that the heuristics can hardly be seen as feasible means of solving
our problem.

Nonetheless, it is interesting to also compare the running times of
the heuristics depending on the number of voters. For the greedy heuristic,
in Figure~\ref{m3_1000_zoom_a.pdf}, we see that for $1000$ voters the
heuristic completes in at most half of the time of the ILP-based
algorithm,  even for the most difficult instances. Yet, for the
case of $10000$ voters (Figure~\ref{m3_10000_zoom_a.pdf}) we already
see that it can be up to 20 times slower than ILP (although for most
instances the heuristic is at most five times slower). This is natural
as the ILP-based solution scales logarithmically with the number of
voters, whereas the greedy heuristic scales polynomially. On the other
hand, for SA we do not see such effect as the number of iterations is
fixed and its running time depends logarithmically on the number of
voters (compare Figures~\ref{m3_1000_zoom_b.pdf},
\ref{m3_10000_zoom_b.pdf}).




From the experiments we conclude that the problem of computing the
cheapest campaign for rigging a given election by influencing a
society graph either is a challenging problem, or our heuristics are
poorly chosen or are poorly optimized. While we were able to prove, using certain ILP techniques, that the problem is FPT with respect to the number
of candidates, there is room for future work, both theoretical and
regarding algorithm engineering, to design algorithms that perform
well in practice.
Indeed, in our experiments we used standard local search algorithms
without significant optimization, and our ILP implementation is using
the standard setting of an off-the-shelf ILP solver (Gurobi 7.5). All experiments were run on two Intel Xeon Gold 6230 - 20 Cores 125W 2.1GHz CPU Processor and 192 GB RAM. 

Let us explain one possible direction that seems especially viable to us.
Intuitively, the reason our ILP formulations are difficult to solve is that we are using a geometric tool (ILP) to express complex logical constraints.
Recently, an alternative method for expressing ILP-definable voting rules and diffusion processes was shown via Presburger Arithmetic~\cite{KouteckyT:2020}.
There exist Presburger Arithmetic solvers such as Omega~\cite{kelly1996omega} and TAPAS~\cite{leroux2009tapas}, and we hope that using them to solve $\mathcal{R}$-BSG would yield interesting results.

\section{Outlook}\label{section:outlook}

We described a powerful model capturing various scenarios of opinion
diffusion in networks, under various manipulative actions. By
considering voter types and society graphs, we were able to provide quite strong tractability results.  In particular, we have shown
that, under certain circumstances, it is possible to find an optimal
bribery scheme, taking into account various diffusion processes
operating on various society graphs, for very general models.
Below we discuss several research directions following from our work.

\begin{description}
\item[ILP Techniques.] We do hope that this paper will have the
  side-effect of popularizing a number of ILP techniques within the
  area of computational social choice. While using Lenstra's algorithm
  to obtain FPT algorithm is already well-known, we have used a number
  of tricks that allow expressing more involved constraints than
  typically found in the ILPs used in this area (even though they are,
  generally, well-known in the area of mathematical programming).  We
  hope that promoting these techniques would lead to discovering
  further voting-related algorithms and further transfer of knowledge
  between mathematical programming and computational social choice.

\item[Generalized Diffusion Processes.] We believe that our concept of
  a \emph{generalized diffusion process} deserves more study. Here we
  mainly cared for identifying whether such generalized diffusion
  processes can be efficiently encoded via linear constraints, but
  studying their further properties, such as finding sufficient
  conditions for convergence, is an intriguing research direction. In
  particular, it would be interesting to explore connections between
  generalized diffusion processes and iterative voting (for more
  details on iterative voting, see the work of Meir et
  al.~\cite{iterative-voting} and many papers that followed up on its
  ideas, in particular those of Sina et al.~\cite{sina2015adapting}
  and Tsang and Larson~\cite{manipulation-social-network}). Under
  iterative voting, all the voters observe the votes currently cast by
  all the other voters and, in each round, they can modify their votes
  to obtain a more desirable outcomes. Generalized diffusion processes
  are capable of encoding such dynamics and, additionally, can impose
  restrictions on which votes the voters see (e.g., according to a
  given social network).

\item[Probabilistic Models.] Our model is inherently
  deterministic. While in the introduction we mentioned that, on the
  one hand, such determinism is quite common and, on the other hand,
  there are workarounds to simulate stochastic bahavior, it would be
  quite interesting to build a stochastic ingredient into the model
  in a way that does not require workarounds. Doing so deserves a careful study.

\end{description}

Our work is primarily theoretical, but we have also included an
experimental component. So far, our conclusion from these experiments
is that our bribery problem on society graphs is either quite
challenging to solve or our approaches to solve it are too naive. Thus it is natural to seek better algorithms and to perform experiments on more realistic data, including data coming from real-life settings.

\section*{Acknowledgments}
Piotr Faliszewski was partially supported by the funds of Polish
Ministry of Science and Higher Education assigned to the AGH
University of Science and Technology.  Martin Koutecký was supported
by a Technion postdoctoral fellowship funded by the Israel Science Foundation grant 308/18, by Charles University project UNCE/SCI/004 and by the project 19-27871X of GA ČR.
Nimrod Talmon was supported by the Israel Science Foundation (ISF;
Grant No. 630/19).  We are very grateful to the IJCAI reviewers for
their useful feedback on the early version of this paper.

\ifarxiv
\bibliographystyle{plain}
\else
\bibliographystyle{named}
\fi
\bibliography{bib}

\end{document}